%% file: main.tex
\documentclass[aps,prl,reprint,superscriptaddress,nofootinbib,nobibnotes,floatfix]{revtex4-2}
\usepackage[T1]{fontenc}
\usepackage{lmodern}
\usepackage{amsmath,amssymb,bm,mathtools,graphicx,booktabs,array,tabularx}
\usepackage{tikz}
\usetikzlibrary{arrows.meta,positioning,calc,fit}
\usepackage[colorlinks=true,linkcolor=blue!50!black,citecolor=blue!50!black,urlcolor=blue!50!black]{hyperref}
\newtheorem{theorem}{Theorem}
\newtheorem{maincorollary}[theorem]{{Corollary}}

\newcommand{\ket}[1]{\lvert #1\rangle}

\newcommand{\Tr}{\operatorname{Tr}}

\newcommand{\SU}{\mathrm{SU}}

\newcommand{\head}[1]{\par\smallskip\noindent\emph{#1.---}}
\newif\ifSupplementContents

\newtheorem{wlemma}{Lemma}[section]
\newtheorem{wtheorem}[wlemma]{Theorem}
\newtheorem{wchangedtheorem}[wlemma]{{Theorem}}
\newtheorem{wcorollary}[wlemma]{Corollary}
\newtheorem{wproposition}[wlemma]{Proposition}
\newtheorem{wdefinition}[wlemma]{Definition}
\newtheorem{wremark}[wlemma]{Remark}
\newtheorem{wexamplebody}[wlemma]{Example}
\newenvironment{wexample}[1][]{%
  \if\relax\detokenize{#1}\relax
    \begin{wexamplebody}\normalfont
  \else
    \begin{wexamplebody}[#1]\normalfont
  \fi
}{\end{wexamplebody}}
\newenvironment{wproof}[1][Proof]{\par\noindent\textit{#1.}\ }{\hfill$\square$\par}
\newcommand{\nbfile}{\texttt{circuit.nb}}
\allowdisplaybreaks[3]
\begin{document}
\title{A Catalytic Algorithm for {Transposing} Unknown Qutrit Unitaries\\ {with Tight Query Bounds}}
\author{Yu-Ao Chen}
\email{chenyuao@buaa.edu.cn}
\affiliation{School of Software, Beihang University}
\date{{October 9}, 2026}
\hypersetup{pdftitle={A Catalytic Algorithm for Transposing Unknown Qutrit Unitaries with Tight Query Bounds},pdfauthor={Yu-Ao Chen}}
\begin{abstract}

We construct exact deterministic circuits for transposing and inverting an unknown qutrit unitary $U\in\mathrm{SU}(3)$ using only forward queries. Three queries prepare a catalyst independently of the input state. Four further queries implement transposition while returning the catalyst unchanged and uncorrelated with the output. Two uses within an antisymmetric encoding implement inversion with eight processing queries and return the same catalyst. Reusing the catalyst implements $m$ transpositions or inversions with $3+4m$ or $3+8m$ total queries, respectively. More generally, we prove that exact transposition and inversion of an unknown $U\in\mathrm{SU}(d)$, $d\ge2$, require at least $d+1$ and $d^2-1$ processing queries, respectively, in fixed causal order. These processing bounds hold for any program state prepared independently of the input state using finitely many forward queries, even if that state is consumed during processing. Both our catalytic qutrit algorithms attain their respective bounds. Moreover, exact qutrit transposition with four processing queries requires at least three preparation queries, and our construction attains this bound.
\end{abstract}
\maketitle

{Forward access to an} unknown quantum evolution {lets} {us} {transform} {it
by}{ }{interleaving oracle calls with known operations}. Quantum supermaps{
and} combs {describe }{this}{ }{use of an oracle}{~}\cite{supermaps,combs,networks2009}{.}{ }{Two basic examples are unitary transposition, $U\mapsto U^\top$, and inversion, $U\mapsto U^\dagger$.} {Heralded}{
}{universal} protocols{~}\cite{quintino2019prl,quintino2019pra},
deterministic sequential {transformations~}\cite{quintinoebler}{,} exact
qubit inversion{~}\cite{yoshida}{,} and the quantum unitary reversal
algorithm (QURA)~\cite{qura} {achieve these tasks in}{ different ways}{.}{
}

{{We store information about the oracle in a program state prepared using finitely many forward queries to $U$, independently of the {state to be transformed}, as in quantum learning and storage-and-retrieval protocols~\cite{bisio2010,sedlak2019}. We distinguish these preparation queries from the processing queries, which are used with the stored state to implement the target transformation.}
If {processing returns} the {state} unchanged and uncorrelated with the output, {it acts as a catalyst and }its preparation cost is paid only
once.}{ This role is shared by }{returned }{resources}{ }{in}{ entanglement}{ }{catalysis~\cite{jonathan1999},}{
}catalytic {magic-state}{ }{conversion~\cite{campbell2011,gidney2019}}, {and}{
}{catalytic} qubit{ }inversion~\cite{yoshida}. {}

{We give an analytic catalytic realization of the seven-query cost numerically identified for qutrit transposition in Ref.~\cite{grinko}.
Three forward queries prepare a catalyst, and four further queries implement
$U^\top$ while returning {the catalyst} unchanged. This processing cost is optimal even
{if processing} may {consume the} program{ state}. At the four-query threshold,
at least three preparation queries are necessary, and our construction attains
this bound. Two uses of the circuit within an antisymmetric encoding implement
$U^\dagger$, also at the optimal processing cost. Reusing the catalyst gives
total query costs $3+4m$ for $m$ transpositions and $3+8m$ for $m$ inversions.}

{We}{ }{obtain}{ these {query costs}}{ by making transposition a fixed operation on
encoded {input states}}{.}{ }{Two queries encode the {system qutrit and catalyst} into a representation space}{;
}a {fixed}{ flip }{then}{ }{produces}{ the corresponding state}{ }{for $U^\dagger$}{
}and {the input state transformed by $U^\top$. }{We}{ }{invert the}{ encoder }{with }{two forward queries}{
}{to $U$}{ and}{ restore the catalyst with}{ a final fixed exchange}{.}{
}{Covariant-network and representation methods }{support this
construction~}{\cite{covariant2009,quintinoebler,grinko}}{, with}{ the}{
recouplings, weights, and fixed gates}{ given in the Supplemental
Material~\cite{supplement}. } 

\head{Model and catalyst}
Let $U\in\SU(3)$ be an unknown oracle. Transpose, entrywise conjugation, and conjugate
transpose are denoted by $\top$, $*$, and $\dagger${ in a fixed basis}.
We allow forward queries, known operations, and arbitrary finite
ancillary systems in a fixed causal order. {Preparation uses $p$ {forward} queries {to $U$ }{to produce a program state}{ from a fixed initial state}, {independently of the state to be transformed}. Processing {acts on the state to be transformed and the program state, using} $q$ further queries to implement $f(U)$, where
$f(U)=U^\top$ or $U^\dagger$.} We write $q_\top$ and $q_\dagger$ for the processing counts of transposition and inversion, respectively. {The protocol implements the unitary channel associated with $f(U)$ exactly on every input state.}

{The catalyst combines a singlet with the adjoint representation.
}{Choose a Hilbert--Schmidt orthonormal basis} $T_{0},\ldots,T_{7}$ {of
traceless} Hermitian {$3\times3$ matrices} with algebraic entries and {define the $8\times8$ matrix $R(U)$ by its entries}
\begin{equation}
 R(U)_{{jk}}=\Tr(T_{{j}}^\dagger UT_{{k}}U^\dagger){,\qquad 0\le {j,k}\le7}.
 \label{eq:adjoint}
\end{equation}
{This real orthogonal} matrix represents conjugation on the
traceless-operator space{. With
$|R(U)\rangle\!\rangle=\sum_{{j,k}=0}^7 R(U)_{{jk}}|{j}\rangle\otimes|{k}\rangle$},
{the program space} and {catalyst are}
\begin{equation}
 |c(U)\rangle_{\mathsf P}=\frac1{\sqrt5}\oplus
 \frac{|R(U)\rangle\!\rangle}{\sqrt{10}}\mathrel{\in}\mathcal H_{\mathsf P}
 =\mathbb C\oplus(\mathbb C^8\otimes\mathbb C^8).
 \label{eq:catalyst}
\end{equation}
{The scalar first component carries the trivial representation}{, and}{
}{$\langle\!\langle R(U)|R(U)\rangle\!\rangle=8$}{ }{gives}{ the }{normalization}. {We embed this} 65-dimensional {program space as the first
summand of a 90-dimensional memory
$\mathcal H_{\mathsf A}=\mathcal H_{\mathsf P}\oplus\mathbb C^{25}$,
and use the same ket with subscript $\mathsf A$} for the {embedded state}.
{The system qutrit $\mathsf S$, on which the target transformation acts,
has state space $V=\mathcal H_{\mathsf S}=\mathbb C^3$.} {The circuit} acts on $\mathcal H=\mathcal H_{\mathsf S}\otimes\mathcal H_{\mathsf A}$,
and each query is $U_{\mathsf S}\otimes I_{\mathsf A}$.

\begin{theorem}[Exact catalytic transposition]
\label{thm:construction}
The catalyst $\ket{c(U)}_{{\mathsf A}}$ {defined by} Eq.~\eqref{eq:catalyst}
can be prepared with three forward queries {to $U$}.
There exists a deterministic {circuit} ${C}_4(U)$, built from
$U$-independent gates and four forward queries {to $U$}, such that
\begin{equation}
 C_4(U)\bigl(I_{\mathsf S}\otimes\ket{c(U)}_{\mathsf A}\bigr)
 =U_{{\mathsf S}}^\top\otimes\ket{c(U)}_{\mathsf A}
 \label{eq:catalytic}
\end{equation}
for every $U\in\SU(3)$.
For every positive integer $m$, reusing this catalyst implements
$m$ transpositions with $3+4m$ {queries in total, including catalyst preparation}.
\end{theorem}

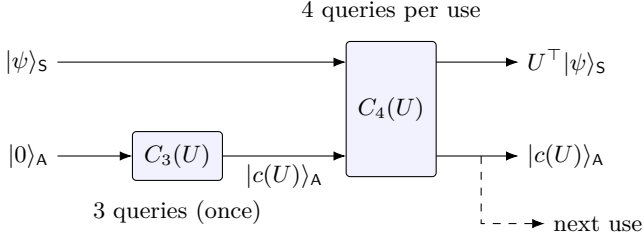
\begin{figure}[!ht]
\centering
\resizebox{\columnwidth}{!}{%
\begin{tikzpicture}[x=1cm,y=1cm,>=Latex,
 gate/.style={draw,rounded corners=1.5pt,fill=blue!5,align=center}]
 \draw[->] (0,.7)--(3.85,.7);
 \node[left] at (0,.7) {$\ket\psi_{\mathsf S}$};
 \draw[->] (0,-.55)--(1.0,-.55);
 \node[left] at (0,-.55) {$\ket0_{{\mathsf A}}$};
 \node[gate,minimum width=1.2cm,minimum height=.65cm] (prep) at (1.6,-.55)
 {$C_3(U)$};
 \draw[->] (2.2,-.55)--(3.85,-.55);
 \node[below,font=\small] at (3.025,-.55) {$\ket{c(U)}_{{\mathsf A}}$};
 \node[gate,minimum width=1.2cm,minimum height=1.8cm] at (4.45,.075)
 {${C}_4(U)$};
 \draw[->] (5.05,.7)--(6.15,.7);
 \draw[->] (5.05,-.55)--(6.15,-.55);
 \node[right] at (6.15,.7) {$U^{{\top}}\ket\psi_{\mathsf S}$};
 \node[right] at (6.15,-.55) {$\ket{c(U)}_{{\mathsf A}}$};
 \node[below=4pt,font=\small] at (prep.south)
 {3 queries (once)};
 \node[above,font=\small] at (4.45,1.13) {4 queries per use};
 \draw[->,dashed] (5.65,-.55)--(5.65,-1.45)--(6.5,-1.45);
 \node[right,font=\small] at (6.5,-1.45) {next use};
\end{tikzpicture}}
\caption{An exact reusable catalyst. Three forward queries prepare
$\ket{c(U)}_{{\mathsf A}}${ from the known memory state $\ket0_{\mathsf A}$}; each four-query {circuit} returns it in a product with the
{output state}. The {returned} catalyst {is} {ready} {for} {reuse}.}
\label{fig:transpose}
\end{figure}

\head{Preparation and processing}
The antisymmetric {encoding}{ }{supplies}{{ }}{the conjugate oracle needed to
prepare the catalyst. Independently of the input state, we apply the}{{ fixed matrix }}{of}{{ }}{the isometry
}{${A_{\bar3}^{3,3}:}{ V^*}\to{ V}\otimes{ V}$}{{ to an auxiliary qutrit in its}}{
}{{computational basis}}{,}{{ as in Eq.~\eqref{w4:eq:iota}}}{. It}{
obeys}~\cite{miyazaki}
\begin{equation}
 {\bigl(A_{\bar3}^{3,3}\bigr) }^\dagger {A_{\bar3}^{3,3} }=I_3,\qquad (U\otimes U){A_{\bar3}^{3,3} }=\det(U){A_{\bar3}^{3,3} }U^*.
 \label{eq:wedge}
\end{equation}
{Thus two forward queries implement $U^*$, while a third supplies $U$.
}{Using}{ $U\otimes U^*\simeq1\oplus R(U)$}{,}{ }{we}{ }{prepare the catalyst in}{
Eq.~\eqref{eq:catalyst} from its}{ value at $U=I$ by acting on}{ }{one half}{
}{of the adjoint Bell state~\cite{hayashi2026}. }The fixed encodings and decodings are given in Sec.~\ref{w4:sec:catalyst-preparation} of the Supplemental Material.{
}
{For processing, }{we construct fixed unitaries $G_0,G_1,G_2$ such that the two-query encoder
}{\begin{equation}
 \mathcal E(U)=G_2(U_{\mathsf S}\otimes I_{\mathsf A})G_1(U_{\mathsf S}\otimes I_{\mathsf A})G_0
 \label{eq:encoder-main}
\end{equation}
}{satisfies
\begin{equation}
 \mathcal E(U)(I_{\mathsf S}\otimes|c(U)\rangle_{\mathsf A})
 =A_{\mathcal M}^{\mathcal H}\Phi(U).
 \label{eq:encoder-endpoint-main}
\end{equation}
}{The fixed map}{ $A_{\mathcal M}^{\mathcal H}$ embeds the}{
45-dimensional intermediate carrier
}{\begin{equation}
 \mathcal M=({V}\otimes\operatorname{Sym}^2{V})\oplus(\operatorname{Sym}^2{(}{V}^*{)}\otimes {V}^*)\oplus({V}\otimes {V}^*)
 \label{eq:middle-main}
\end{equation}
}{into the workspace{ $\mathcal H$}, where ${V^*}$ is {the conjugate Hilbert space of $V$}. }The three summands allow a fixed flip to exchange the representation actions and apply $U^\top$ to the input state, as in Eq.~\eqref{eq:flip-main}. {The map $\Phi(U)$ is an
isometry from the system space{ $\mathcal H_{\mathsf S}$} to $\mathcal M$.
}
{On $\mathcal M$ we choose a} fixed unitary flip $J${, whose} physical
extension $\widehat J$ satisfies
{\begin{equation}
 \widehat J A_{\mathcal M}^{\mathcal H}\Phi(U)=A_{\mathcal M}^{\mathcal H}\Phi(U^\dagger)U^\top.
 \label{eq:flip-main}
\end{equation}
}{The flipped map is therefore the encoder output at $U^\dagger$ for the input transformed by $U^\top$ and catalyst $|c(U^\dagger)\rangle_{\mathsf A}$.
}{The inverse encoder is
}{\begin{equation}
 \mathcal E(U^\dagger)^\dagger=G_0^\dagger(U_{\mathsf S}\otimes I_{\mathsf A})G_1^\dagger(U_{\mathsf S}\otimes I_{\mathsf A})G_2^\dagger,
 \label{eq:reverse-main}
\end{equation}
}{which requires two further forward queries. }{Applying the inverse encoder to the right-hand side of Eq.~\eqref{eq:flip-main}
and using Eq.~\eqref{eq:encoder-endpoint-main} with $U$ replaced by $U^\dagger$ gives}
{\begin{equation}
 \mathcal E(U^\dagger)^\dagger A_{\mathcal M}^{\mathcal H}\Phi(U^\dagger)U_{\mathsf S}^\top
 =U_{\mathsf S}^\top\otimes|c(U^\dagger)\rangle_{\mathsf A}.
 \label{eq:uncompute-main}
\end{equation}}
{T}he fixed {swap gate }$G_{{\mathrm{swap}}}=(1\oplus\mathrm{SWAP}_{8,8})\oplus I_{25}$ {satisfies }${G_{\mathrm{swap}}|c(U^\dagger)\rangle_{\mathsf A}}\mathrel{{=}}|c(U)\rangle_{\mathsf A}$. {Combining these steps gives the} complete circuit
{\begin{equation}
 C_4(U)=(I_{\mathsf S}\otimes G_{{\mathrm{swap}}})
 \mathcal E(U^\dagger)^\dagger\widehat J\mathcal E(U).
 \label{eq:closure-main}
\end{equation}
}{Equations~\eqref{eq:flip-main}--\eqref{eq:closure-main} establish the
catalytic identity with four processing queries. }
\head{Transpose covariance}
{To design one circuit that works for every unknown $U$, we use transpose
covariance to relate the intermediate maps at different values of $U$.
For $g,h\in\SU(3)$, transposition obeys
$(gUh^\top)^\top g^*=hU^\top$; we extend this relation to larger output spaces.
Let the output space $\mathcal X$ carry a unitary representation $\rho$ of
$\SU(3)\times\SU(3)$. We call a family of linear maps
$K(U):\mathcal H_{\mathsf S}\to\mathcal X$ transpose covariant if}
\begin{equation}
 K(gUh^\top)g^*=\rho(g,h)K(U)
 \label{eq:covariance-main}
\end{equation}
for {all }$g,h{,U}\in\SU(3)${. This relation determines the family from its seed $K(I)$ through }$K(U)=\rho(U,I)K(I)U^\top${. Thus }two {transpose-covariant families} with the same output
representation agree for every $U$ if their seeds agree.
{We use}{ }{Clebsch--Gordan decompositions}{~\cite{grinko}} {to}{ }{identify the representation}{
}{branches}{ }{reached by}{ a query}{,}{ }{apply tensor--Hom duality~\cite{etingof2011} to select those admitting nonzero seeds, and then match }{their}{ seeds }{through}{ }{two
recoupling}{ }{identities}{. }{For this three-block midpoint, the first fixes the}{ }{singlet--adjoint weight}{
}{$1:4$}{ and permits a fixed}{ }{rotation of two equivalent representation}{
}{copies; the second produces}{ }{$\Phi$. }{These}{ }{encodings occupy
memories of dimensions 90 and 30 at the}{ first and second }{queries}{ and}{
}{extend to unitaries on the common}{ workspace}{
}{(Secs}{.}{~\ref{w4:chap:covariance} and~\ref{w4:chap:processing}}{ }{of the}{
}{Supplemental Material}{)}{.}{
}
\head{Reuse and inversion}
{Equation}~\eqref{eq:catalytic} returns {the catalyst unchanged} and {uncorrelated with} the output. {The circuit} can {therefore be applied to} another input state or to the previous output {without a} new preparation.
{Repeating it} $m$ {times uses} $3+4m$ queries {in total}.

\begin{maincorollary}[Exact catalytic inversion]
\label{thm:inversion}
For every $U\in\SU(3)$, two successive uses of the {circuit} in
Theorem~\ref{thm:construction}, combined with fixed antisymmetric encoding
and decoding, implement $U^\dagger$ on any qutrit input while returning
$\ket{c(U)}_{{\mathsf A}}$ unchanged and uncorrelated with the output.
For every positive integer $m$, {reusing this catalyst implements} $m$ inversions {with} $3+8m$ {queries in total},
including catalyst preparation.
\end{maincorollary}

\begin{figure}[!ht]
\centering
\resizebox{\columnwidth}{!}{%
\begin{tikzpicture}[x=1cm,y=1cm,>=Latex,
 gate/.style={draw,rounded corners=1pt,fill=blue!5}]
 \draw (0,-.7)--(8.6,-.7);
 \node[left] at (0,-.7) {$\ket{c(U)}_{{\mathsf A}}$};
 \node[right] at (8.6,-.7) {$\ket{c(U)}_{{\mathsf A}}$};
 \draw (0,.7)--(.35,.7);
 \node[left] at (0,.7) {$\ket\psi_{\mathsf S}$};
 \node[gate,minimum width=.90cm,minimum height=1.35cm]
 at (.8,.7) {${A_{\bar3}^{3,3} }$};
 \draw (1.25,1.15)--(3.45,1.15);
 \draw (1.25,.25)--(3.45,.25);
 \node[above=3pt,font=\small] at (1.525,1.15) {$\mathsf S_1$};
 \node[above=3pt,font=\small] at (1.525,.25) {$\mathsf S_2$};
 \draw (3.45,1.15)--(3.95,.25);
 \draw (3.45,.25)--(3.95,1.15);
 \draw (3.95,1.15)--(5.95,1.15);
 \draw (3.95,.25)--(5.95,.25);
 \draw (5.95,1.15)--(6.45,.25);
 \draw (5.95,.25)--(6.45,1.15);
 \draw (6.45,1.15)--(6.70,1.15);
 \draw (6.45,.25)--(6.70,.25);
 \node[gate,minimum width=1.20cm,minimum height=1.4cm,align=center]
 at (2.4,-.225) {${C}_4(U)$};
 \node[gate,minimum width=1.20cm,minimum height=1.4cm,align=center]
 at (4.95,-.225) {${C}_4(U)$};
 \node[gate,minimum width=1.20cm,minimum height=1.35cm]
 at (7.3,.7) {${\bigl(A_{\bar3}^{3,3}\bigr) }^\dagger$};
 \draw[->] (7.90,.7)--(8.6,.7);
 \node[right] at (8.6,.7) {$U^\dagger\ket\psi_{\mathsf S}$};
 \node[below,font=\small] at (2.4,-1.02) {4 queries};
 \node[below,font=\small] at (4.95,-1.02) {4 queries};
\end{tikzpicture}}

\caption{Exact catalytic inversion.
Two four-query uses of ${C}_4(U)$,
combined with the fixed {isometric }encoding ${A_{\bar3}^{3,3} }$, SWAP gates, and decoding ${\bigl(A_{\bar3}^{3,3}\bigr) }^\dagger$,
apply $U^\dagger$ to the {system qutrit} and return $\ket{c(U)}_{{\mathsf A}}$ in a
product with the output.
The catalyst {is} {ready for reuse}.}
\label{fig:inverse}
\end{figure}
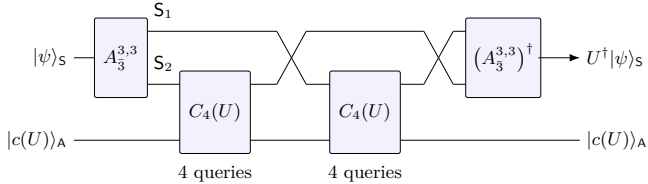

{Applying Eq.~\eqref{eq:wedge} to $U^\top$ gives}
\begin{equation}
 {\bigl(A_{\bar3}^{3,3}\bigr) }^\dagger(U^{{\top}}\otimes U^{{\top}}){A_{\bar3}^{3,3} }
       =(U^{{\top}})^*=U^\dagger .
 \label{eq:inverse}
\end{equation}
{We realize the two factors using successive calls to $C_4(U)$ with the same catalyst, between the fixed encoding and decoding.}
{The first inversion therefore uses} eleven forward queries, {including catalyst preparation, which is one fewer than the twelve required by the} modified QURA construction~\cite[Appendix B.3, Theorem S13]{qura}.

\head{Matching resource bounds}
{Our construction attains the minimum processing costs for exact qutrit transposition and inversion: four and eight forward queries, respectively. {We obtain these values from processing bounds valid in arbitrary dimension, even when the prepared program state is consumed. For qutrit transposition, attaining the four-query processing bound also fixes the minimum preparation cost at three queries.}}

\begin{theorem}[{Processing lower bounds}]
\label{thm:bounds}
{Let $d\ge2$ and let $U\in\SU(d)$ be unknown. Consider a deterministic protocol in a fixed causal order. Its preparation stage uses any finite number $p$ of forward queries to $U$ to produce a program state from a fixed initial state, independently of the state to be transformed. Its processing stage uses {the prepared program state and }$q$ further forward queries and implements the unitary channel associated with $f(U)$ exactly for every $U$, where $f(U)=U^\top$ or $U^\dagger$. All other operations are known and independent of $U$. Then}
\begin{equation}
             q_{\top }\ge d+1,\qquad q_{{\dagger}}\ge d^2-1 .
 \label{eq:rate}
\end{equation}
{These bounds allow arbitrary finite-dimensional ancillary systems, consumption of the program state, and a final unobserved environment that depends on $U$.}
\end{theorem}

{For $d=3$, Theorem~\ref{thm:construction} and Corollary~\ref{thm:inversion} attain these processing bounds while returning the catalyst.}

{We extend the differential completely positive (CP) method of Ref.~\cite{odake} to finitely prepared program states that may be consumed. Holding the input program state fixed while differentiating the processor gives a CP response $\Lambda$ with $\Lambda(I)=qI$. Differentiating the exact input-output identity also retains the derivatives of the input program state and final environment; these contribute only a scalar multiple of the system identity. For traceless Hermitian $H$, the response is therefore $H^\top+\Tr(FH)I$ for transposition or $-H+\Tr(FH)I$ for inversion, with a Hermitian $F$ chosen separately for each target. Choi positivity gives the stated bounds (Sec.~\ref{sm:consumable} of the Supplemental Material).}

{At the qutrit transposition threshold, saturation of the processing bound {yields the following} preparation{ bound, whose proof is given} in Sec.~\ref{sm:generalstorage} of the Supplemental Material.}

{\begin{theorem}[Preparation lower bound]
\label{thm:preparation}
In the model of Theorem~\ref{thm:bounds}, every exact qutrit transposition protocol with $q_\top=4$ requires $p\ge3$ preparation queries, including when processing consumes the program state. The construction in Theorem~\ref{thm:construction} attains $p=3$ at $q_\top=4$.
\end{theorem}}

{\head{Verification}}
{We verify our quantum circuit for $U\in\SU(3)$ symbolically in
Mathematica, checking the unitarity of the fixed gates and the three-query
preparation. We extend the circuit expressions
to arbitrary complex $3\times3$ matrices and use $|c(U)\rangle_{\mathsf A}$
for the degree-three memory column produced by preparation. Treating the
nine matrix entries as independent commuting variables, we verify
\begin{equation}
C_4(U)\bigl(I_{\mathsf S}\otimes|c(U)\rangle_{\mathsf A}\bigr)=\det(U)\bigl[U^\top\otimes|c(U)\rangle_{\mathsf A}\bigr].
\label{eq:notebook-main}
\end{equation}
Both sides are $270\times3$ polynomial matrices, and all 810 entry
differences vanish exactly. On $\SU(3)$, the prepared column equals the
catalyst of Eq.~\eqref{eq:catalyst}, and reduction modulo $\det U-1$ gives
Eq.~\eqref{eq:catalytic}. Section~\ref{sm:symbolic} gives the explicit
polynomial extension and reproducible checks.}
{We also formalize the exact constructions and lower bounds
of Secs.~\ref{w4:chap:construction}--\ref{sm:generalstorage}, together with
the repeated-use constructions, in Lean 4 (Sec.~\ref{sm:lean}); the bounds
apply to the finite causal protocol models specified there.}

\head{Discussion}
{The reversible processing cycle attains the optimal per-use processing query cost in this model
while returning the program state after every invocation.}
{Each subsequent invocation combines the returned state with fresh queries to $U$, without repeating preparation.}
Returning the program state therefore carries no additional processing-query cost for exact qutrit transposition and inversion.

The{ }{remaining}{ }{resource questions}{ concern }{single-use}{ }{queries and}{
}{memory.}{ }{Our}{ first transposition }{uses} seven queries and {our}{ first}{
inversion} {eleven}. {For}{ unrestricted} single-use{ protocols}{, }{where} the{
}{input state }{can} participate from the first query{, the} analytic lower bounds{
}{of Ref.~\cite{odake} }{are} six {and nine} queries{ }for {the}{ respective}{
}qutrit {tasks. }{Our preparation bound applies when transposition uses four processing queries and preparation acts independently of the state to be transformed. }{The construction gives $M_{\min}\le90$, where $M_{\min}$ is the smallest
peak memory dimension outside the query qutrit {throughout preparation and processing} among pure
coherent, deterministic protocols in a fixed causal order that prepare the catalyst with three queries, independently of the state to be transformed, and use four processing queries to
transpose an unknown qutrit unitary exactly and return the catalyst.
The peak includes the logical system whenever stored outside the query
qutrit. Exact coefficient-Gram computations suggest that
$M_{\min}=90$ in this model; a matching memory lower bound
remains open (Sec.~\ref{sm:auxiliary}).}

\begin{acknowledgments}
\head{Acknowledgments}
This work was supported by the Fundamental Research Funds for the Central
Universities (Grant No. KG16-4773-01).
{The author used OpenAI Codex (ChatGPT 6 Astra) {for assistance} with mathematical derivations, {computational and Lean verification, and} manuscript preparation{.}}
\end{acknowledgments}
{\head{Data availability}
}{The Mathematica notebooks, Python scripts, and Lean sources used in this work are included in the accompanying ancillary files.}

\makeatletter
\def\bibsection{%
  \par
  \begingroup\baselineskip14\p@
    \bib@device{\hsize}{72\p@}%
  \endgroup
  \nobreak\@nobreaktrue\addvspace{7\p@}%
}
\makeatother

\clearpage
\onecolumngrid
\input{supplemental.tex}

\end{document}

%% file: supplemental.tex
\setcounter{section}{0}
\setcounter{equation}{0}
\setcounter{table}{0}
\setcounter{secnumdepth}{3}
\renewcommand{\thesubsubsection}{\thesubsection.\arabic{subsubsection}}
\renewcommand{\thesection}{S\arabic{section}}
\renewcommand{\thesubsection}{\thesection.\arabic{subsection}}
\makeatletter
\renewcommand{\p@subsection}{}
\renewcommand{\p@subsubsection}{}
\makeatother
\renewcommand{\theequation}{S\arabic{equation}}
\renewcommand{\thetable}{S\arabic{table}}
\renewcommand{\theHsection}{supp.\arabic{section}}
\renewcommand{\theHequation}{supp.\arabic{equation}}
\renewcommand{\theHtable}{supp.\arabic{table}}
\newcommand{\tr}{\Tr}
\begin{center}
\textbf{\large Supplemental Material}\par\medskip
\textbf{A Catalytic Algorithm for {Transposing} Unknown Qutrit Unitaries {with Tight Query Bounds}}\par\medskip
Yu-Ao Chen\par
School of Software, Beihang University
\end{center}
We give a complete analytic proof of the four-query catalytic transposition algorithm in the first three sections, from catalyst preparation and transpose covariance to the explicit representation construction.
We then establish the optimal processing query costs in Sec.~\ref{sm:consumable} and the optimal preparation query cost for four-query transposition in Sec.~\ref{sm:generalstorage}.
Sec.~\ref{sm:mathematica} presents exact symbolic verification of the complete circuit and formal verification {of Secs.~\ref{w4:chap:construction}--\ref{sm:generalstorage} and the repeated-use constructions} in Lean 4.
It also gives computational evidence on the minimum auxiliary dimension for pure coherent transposition of an unknown qutrit unitary with three preparation queries and four processing queries.

\par\medskip\noindent\textbf{Contents}\par\smallskip
\begingroup
\setcounter{tocdepth}{2}
\SupplementContentsfalse
\makeatletter
\def\l@section{\@dottedtocline{1}{0em}{2.5em}}
\def\l@subsection{\@dottedtocline{2}{2.5em}{3.5em}}
\def\l@subsubsection{\@dottedtocline{3}{6em}{4.5em}}
\def\@pnumwidth{2em}
\def\@tocrmarg{3em}
\let\SupplementContentsLine\contentsline
\renewcommand{\contentsline}[4]{%
  \ifSupplementContents\SupplementContentsLine{#1}{#2}{#3}{#4}\fi}
\@starttoc{toc}
\makeatother
\endgroup
\addtocontents{toc}{\protect\SupplementContentstrue}
\medskip

\clearpage
\section*{Summary of notation}
The table collects recurring symbols and basic conventions. Symbols used only in a local derivation are defined where they first occur. The last column points to defining equations or passages in the supplement; grouped symbols may have several entries.
\par\medskip
\begingroup
\small\renewcommand{\arraystretch}{1.18}\setlength{\tabcolsep}{5pt}\noindent
\begin{tabularx}{\textwidth}{@{}>{\raggedright\arraybackslash}p{0.29\textwidth}>{\raggedright\arraybackslash}X>{\raggedright\arraybackslash}p{0.14\textwidth}@{}}
\toprule
\textbf{Symbol} & \textbf{Meaning and convention} & \textbf{Definition}\\
\midrule
\(U;\ g,h\) & Unknown unitary {$U\in SU(3)$} and left/right covariance variables {$g,h\in SU(3)$}. The {processing bound allows $U\in SU(d)$, $d\ge2$}. & Secs.~\ref{w4:sec:target}--\ref{w4:sec:notation}{; Sec.~\ref{sm:consumable}}\\
\(V,\ V^*\) & Fundamental carrier $V=\mathbb C^3$ and its conjugate carrier, with actions $g$ and $g^*$, respectively. & Secs.~\ref{w4:sec:target}--\ref{w4:sec:notation}\\
\(\mathcal H_{\mathsf S},\mathcal H_{\mathsf P}, \mathcal H_{\mathsf A},\mathcal H\) & {Input state space}, program{ space}, full processing memory, and joint workspace; dimensions $3,65,90,270$ in the {fixed $a=4/5$} construction. {In Secs.~\ref{sm:consumable}--\ref{sm:generalstorage}, $\mathcal H_{\mathsf P}$ is an arbitrary finite ancillary space; Sec.~\ref{sm:consumable} uses $\dim\mathcal H_{\mathsf S}=d$.} & Sec.~\ref{w4:sec:target}; Eq.~\eqref{w4:eq:workspace}{; Sec.~\ref{sm:consumable}}\\
\(\mathsf S,\mathsf P,\mathsf A; I_d,I_{\mathsf S},I_{\mathsf A}\) & Upright labels name physical systems. Identity subscripts specify a dimension or a system. & Sec.~\ref{w4:sec:notation}\\
\({}^\top,\ {}^*,\ {}^\dagger{;\ \operatorname{adj}M}\) & Transpose, entrywise conjugation in the fixed bases, and adjoint. {$\operatorname{adj}M$ is the classical adjugate; in Sec.~\ref{sm:generalstorage}, $M^\dagger$ has this meaning.} & Sec.~\ref{w4:sec:target}{; Eq.~\eqref{sm:encodingfactorization}; Sec.~\ref{sm:symbolic}}\\
\(|M\rangle\!\rangle,\ |I_d\rangle\!\rangle\) & Unnormalized vectorization in the fixed bases{, with the output factor first}; $|I_d\rangle\!\rangle$ is the unnormalized Bell vector. & Sec.~\ref{w4:sec:notation}\\
\(\operatorname{Sym}^2V,\operatorname{Sym}^2(V^*), \operatorname{Sym}^2U\) & Symmetric tensor carriers and the induced representation. $\operatorname{Sym}^2U$ is not the matrix square $U^2$. & Sec.~\ref{w4:sec:notation}\\
\(\mathrm{SWAP}_{d_1,d_2},\mathrm{SWAP}\) & Exchange of ordered tensor factors. In the qutrit construction, $\mathrm{SWAP}$ abbreviates $\mathrm{SWAP}_{3,3}$. & Eq.~\eqref{w4:eq:swap}\\
\(\lambda,\mu;\ W_\lambda,W_\mu{,}W_8; d_\lambda,d_\mu\) & Irreducible representations, their carriers, and their dimensions; representation matrices are written $\lambda(g),\mu(g)$. {$W_8$ is the }{t}raceless-matrix adjoint carrier{.} & Sec{s}.~\ref{w4:tc:sec:blocks}{, }\ref{w4:sec:branches}\\
\(1,3,\bar3,6,\bar6,8,10, \overline{10},15,\overline{15},27\) & Dimension labels denote the particular irreducible representations used here; bars denote conjugates. General representation sums retain Young-diagram labels. & \hyperref[w4:tc:representation-labels]{Notation in Sec.~\ref*{w4:chap:covariance}}\\
\(A_\nu^{\lambda,\mu};\ A_{\lambda,\mu}^{\mathcal X}\) & Fixed isometric embeddings: a Clebsch--Gordan (CG) component $W_\nu\hookrightarrow W_\lambda\otimes W_\mu$, or a left/right block $W_\lambda\otimes W_\mu\hookrightarrow\mathcal X$. Bases and phases are fixed below. & Eq.~\eqref{w4:tc:block-embedding}{; Sec.~\ref{w4:sec:notation}}\\
\(A_{\bar3}^{3,3},A_6^{3,3},A_8^{3,\bar3}\) & Standard embeddings of the antisymmetric, symmetric, and adjoint components, respectively. Conjugate embeddings use the entrywise-conjugate matrices. & Sec.~\ref{w4:sec:notation}; Eq.~\eqref{w4:eq:adjoint-coordinate-embedding}\\
\(\Pi_\lambda,\Pi_{\operatorname{Sym}^2}\) & Orthogonal projections onto the specified irreducible component or symmetric subspace, in the stated ambient carrier. & Secs.~\ref{w4:sec:branches}, \ref{w4:sec:notation}\\
\(T_0,\ldots,T_7\) & Fixed Hilbert--Schmidt orthonormal basis of traceless Hermitian $3\times3$ matrices. & Sec.~\ref{w4:sec:notation}\\
\(X,\ |X\rangle\) & {A traceless matrix $X\in W_8$} and its eight-dimensional coordinate ket in the fixed Hermitian basis. & Sec.~\ref{w4:sec:branches}\\
\(\operatorname{Ad}_g,\ R(g)\) & Adjoint action $X\mapsto gXg^\dagger$ and its real $8\times8$ matrix in the fixed Hermitian basis. & Eq.~\eqref{w4:eq:adjoint-action}; Sec.~\ref{w4:sec:target}\\
\(a,\ |c_a(U)\rangle{;\ |c(U)\rangle}\) & Catalyst weight {$0\le a\le1$} and catalyst family. The main algorithm uses $a=4/5$; the known preparation {seed} is $|c_a(I)\rangle$. {The main text writes $|c(U)\rangle=|c_{4/5}(U)\rangle$ and explicitly extends this ket polynomially in Verification. The compatible four-query family uses $4/5\le a\le1$.} & Eq.~\eqref{w4:eq:catalyst}{; Eq.~\eqref{w4:eq:preparation-seed}; Remark~\ref{w4:rem:general-catalyst-family}}\\
\bottomrule
\end{tabularx}
\endgroup
\clearpage
\noindent\textit{Summary of notation (continued)}
\par\medskip
\begingroup
\small\renewcommand{\arraystretch}{1.18}\setlength{\tabcolsep}{5pt}\noindent
\begin{tabularx}{\textwidth}{@{}>{\raggedright\arraybackslash}p{0.29\textwidth}>{\raggedright\arraybackslash}X>{\raggedright\arraybackslash}p{0.14\textwidth}@{}}
\toprule
\textbf{Symbol} & \textbf{Meaning and convention} & \textbf{Definition}\\
\midrule
\(\mathcal M_1,\mathcal M_2,\mathcal M_3,\mathcal M_{4};\ A_{\mathcal M_4}^{\mathcal H}\) & Intermediate spaces in circuit order {for $a=4/5$}: memories $\mathcal M_1,\mathcal M_3$ have dimensions $90,30$; full carriers $\mathcal M_2,\mathcal M_4$ have dimensions $75,45$. $A_{\mathcal M_4}^{\mathcal H}$ is the fixed isometric embedding of $\mathcal M_4$ into $\mathcal H$. {The main text abbreviates $\mathcal M_4$ and $A_{\mathcal M_4}^{\mathcal H}$ as $\mathcal M$ and $A_{\mathcal M}^{\mathcal H}$.} & Eqs.~\eqref{w4:eq:first-memory}, \eqref{w4:eq:first-occupied-space}, \eqref{w4:eq:second-memory}, \eqref{w4:eq:middle-space}; Sec.~\ref{w4:sec:notation}\\
\(\Phi_{\lambda,\mu}(U)\) & Concrete normalized covariant components. The three basic types are $(3,6)$, $(\bar6,\bar3)$, and $(3,\bar3)$. & Eq.~\eqref{w4:eq:phi-components}\\
\(\Phi_\alpha(U),\ \alpha{;\ \Phi(U)}\) & Weighted direct sum of the three basic components, with $0\le\alpha\le1/2$; the main algorithm uses $\alpha=1/5$. {The main text writes $\Phi(U)=\Phi_{1/5}(U)$.} & Eq.~\eqref{w4:eq:phi-family}\\
\(K(U),\mathcal X,\rho(g,h)\) & General transpose-covariant family, its output space, and the product-group representation on that space. & Def.~\ref{w4:def:tc}\\
\(S;\ S_{\lambda,\mu}\) & General block seed and a specified normalized block seed of type $(\lambda,\mu)$. A general $S$ need not be normalized. & Sec.~\ref{w4:tc:sec:blocks}; Eq.~\eqref{w4:eq:initial-block-seeds}\\
\(\mathcal E(U),\mathcal E_{\mathrm{fix}}, \mathcal E_2\) & Encoding maps: the physical two-query encoder, a fixed linear encoding, {and }the second-query encoding. Isometry or unitarity is stated separately. & Eq.~\eqref{w4:eq:two-query-circuit}; Sec.~\ref{w4:tc:sec:queries}; Eq.~\eqref{w4:eq:second-encoding}\\
\(\mathcal D_1,\mathcal D_2\) & Fixed first- and second-query decoding isometries on the stated representation subspaces. & Eqs.~\eqref{w4:eq:first-decoding}, \eqref{w4:eq:second-decoding}\\
\(G_0,G_1,G_2, G_{3\bar3},G_{{\mathrm{swap}}}\) & Fixed physical gates: the encoder gates, singlet--adjoint basis change, and catalyst restoration. {In Secs.~\ref{sm:consumable}--\ref{sm:generalstorage}, $G_j$ denotes the fixed gates of the general protocol.} & Eq.~\eqref{w4:eq:two-query-circuit}; Secs.~\ref{w4:sec:catalyst-preparation}, \ref{w4:sec:four-query-closure}{; Sec.~\ref{sm:consumable}}\\
\(J,\widehat J\) & Fixed flip on the stated carrier, and its unitary extension to the physical workspace. & Sec.~\ref{w4:tc:sec:flip}; Lemma~\ref{w4:lem:middle-flip}\\
\({C}_4(U),\ {C}_{q}{(U)}\) & \leavevmode{}Specific four-query catalytic processor and general {$q$-query }processing circuit, respectively. & Theorem~\ref{w4:thm:S1}; Eq.~\eqref{sm:consumedidentity}\\
\(p,q,q_\top ,q_{\dagger}\) & Preparation and processing {ordinary forward} query counts; $q_\top $ and $q_{\dagger}$ denote processing counts for transposition and inversion. & Sec.~\ref{sm:consumable}\\
\(|c_{\mathrm{in}}(U)\rangle,\ |c_{\mathrm{out}}(U)\rangle\) & \leavevmode{}Normalized input and output program states in the general lower-bound model; all auxiliary systems{, including retained environments,} are included in $\mathsf P$, and the output need not equal the input. & Eq.~\eqref{sm:consumedidentity}\\
{\(\Lambda_f(U,\cdot);\ F_\top,F_\dagger\)} & {Completely positive processing-query response, with $\Lambda_f(U,I)=qI$. The Hermitian matrices $F_\top,F_\dagger$ represent the scalar program-state response; their traces are $q_\top-1$ and $q_\dagger+1$.} & {Eqs.~\eqref{sm:generalET}, \eqref{sm:generalEI}}\\
{\(L_j(M),\widetilde L_j(M)\)} & {Quadratic $3\times3$ slices of the fixed initial encoding with a program{ state prepared using two queries}, and their linear factorization through the classical adjugate.} & {Eqs.~\eqref{sm:encodingslices}, \eqref{sm:encodingfactorization}}\\
{\(M_{\min}\)} & {Smallest peak dimension outside the query qutrit throughout preparation and processing in the pure coherent fixed-causal $3+4$ catalytic model, including logical data stored outside the query register.} & {Sec.~\ref{sm:auxiliary}}\\
\(\delta_{ij},\epsilon_{ijk}\) & Kronecker delta and Levi--Civita tensor, with $\epsilon_{012}=1$. & Sec.~\ref{w4:sec:branches}\\
\bottomrule
\end{tabularx}
\endgroup
\clearpage

\section{Transposition algorithm and catalyst preparation}\label{w4:chap:construction}
Our construction separates a three-query preparation stage from a four-query processor that applies \(U^\top\) to the data and restores the catalyst for reuse. In Sec.~\ref{w4:sec:target}, we state the algorithm and explain its encode--flip--decode mechanism. Sec.~\ref{w4:sec:notation} introduces the notation and intermediate maps, and Sec.~\ref{w4:sec:catalyst-preparation} constructs the catalyst before the data arrive.

\subsection{Transposition algorithm}\label{w4:sec:target}

Let \(U\in SU(3)\) be the unknown unitary acting on the {system} qutrit,
whose Hilbert space \(V=\mathcal H_{\mathsf S}=\mathbb C^3\) has the
fixed computational basis \(\{|j\rangle:j=0,1,2\}\).
For matrices, superscripts \(\top\), \(*\), and \(\dagger\) denote
transpose, entrywise complex conjugation, and conjugate transpose,
respectively.
Let \(R(U)\) be the real orthogonal \(8\times8\) matrix of the
conjugation action \(X\mapsto UXU^\dagger\) in a fixed Hilbert--Schmidt
orthonormal basis of traceless Hermitian \(3\times3\) matrices.
For this matrix, write
\[
|R(U)\rangle\!\rangle
=\sum_{s=0}^7(R(U)|s\rangle)\otimes|s\rangle,
\]
where both ket factors use the same fixed eight-dimensional coordinate
basis. The program space and the processing memory are
\[
\mathcal H_{\mathsf P}=\mathbb C\oplus(\mathbb C^8\otimes\mathbb C^8),
\qquad
\mathcal H_{\mathsf A}=\mathcal H_{\mathsf P}\oplus\mathbb C^{25}.
\]
Define the catalyst used in the construction by
\begin{equation}\label{w4:eq:catalyst-specific}
|c_{4/5}(U)\rangle_{\mathsf P}
=\frac1{\sqrt5}\oplus\frac1{\sqrt{10}}|R(U)\rangle\!\rangle,
\end{equation}
whose scalar component lies in the trivial carrier \(\mathbb C\).
Write \(|c_{4/5}(U)\rangle_{\mathsf A}\) for its image under the natural
first-summand inclusion
\(\mathcal H_{\mathsf P}\hookrightarrow\mathcal H_{\mathsf A}\).
In the main text we write $|c(U)\rangle$ for $|c_{4/5}(U)\rangle$,
and likewise suppress the subscript in its polynomial extension.
The processor acts on \(\mathcal H_{\mathsf S}\otimes\mathcal H_{\mathsf A}\), the joint Hilbert space of the {system} qutrit and this 90-dimensional memory.

\begin{wtheorem}[Four-query catalytic transposition algorithm]\label{w4:thm:S1}
There exists a circuit \({C}_4(U)\), composed of fixed unitary gates and
four forward processing queries to \(U\), such that {for }every \(U\in SU(3)\){,}
\begin{equation}\label{w4:eq:w4-target}
{C}_4(U)\bigl({I_{\mathsf S}}\otimes|c_{4/5}(U)\rangle_{\mathsf A}\bigr)
=U^{\top}\otimes|c_{4/5}(U)\rangle_{\mathsf A},
\end{equation}
where the catalyst \(|c_{4/5}(U)\rangle_{\mathsf A}\) can be prepared
deterministically using fixed unitary gates and three forward queries to \(U\).
\end{wtheorem}

\noindent\textbf{Proof sketch of Theorem~\ref{w4:thm:S1}.}
We first encode the data together with the catalyst using two forward queries. The fixed flip transforms the encoded {map} \(\Phi_{1/5}(U)\) into \(\Phi_{1/5}(U^\dagger)U^\top\). This {is the encoder map at} \(U^\dagger\){ composed on the right with \(U^\top\)}. We then reverse the encoder at \(U^\dagger\): its query gates become \(U\), so two more forward queries recover {the map \(U^\top\otimes|c_{4/5}(U^\dagger)\rangle_{\mathsf A}\)}. A final swap of the catalyst's two adjoint factors restores \(|c_{4/5}(U)\rangle_{\mathsf A}\). The detailed processing construction appears in Sec.~\ref{w4:chap:processing}, with exact symbolic verification in Sec.~\ref{sm:symbolic}.

\subsection{Notation and intermediate maps}\label{w4:sec:notation}

Let \(g,h\in SU(3)\).
We write \(V^*\) for the conjugate Hilbert space of \(V\),
which carries \(g^*\) in the corresponding conjugate computational
basis.
Vectors in \(V^*\) are written as kets; the inner product identifies
\(V^*\) linearly with the linear dual of \(V\).
For a matrix \(M:\mathbb C^{d_2}\to\mathbb C^{d_1}\), with positive
integers \(d_1,d_2\), define the unnormalized vectorization and its bra by
\[
|M\rangle\!\rangle
=\sum_{k=0}^{d_2-1}(M|k\rangle)\otimes|k\rangle,
\qquad
\langle\!\langle M|=(|M\rangle\!\rangle)^\dagger,
\]
where the second ket \(|k\rangle\) belongs to an independent copy of the fixed input basis, and
\(\langle\!\langle M|M'\rangle\!\rangle=\operatorname{tr}(M^\dagger M')\).
Subscripts specify scalar components in these bases:
\(x_j=\langle j|x\rangle\) for a ket \(|x\rangle\), and
\(M_{jk}=\langle j|M|k\rangle\) for a matrix.
For maps written with multiple tensor indices, adjacent indices form
context-defined multi-indices, commas may distinguish factor groups,
and a final semicolon separates output indices from input columns.
Ket components have no input column, and dummy-index contractions
are written with explicit sums.
Denote the traceless Hermitian basis fixed in Sec.~\ref{w4:sec:target}
by \(T_0,\ldots,T_7\); it also spans the traceless complex
\(3\times3\) matrices. In this basis, the matrix introduced above is
\begin{equation}\label{w4:eq:adjoint-matrix}
R(U)=\Bigl(\sum_{r=0}^7|r\rangle\langle\!\langle T_r|\Bigr)
(U\otimes U^*)
\Bigl(\sum_{s=0}^7|T_s\rangle\!\rangle\langle s|\Bigr).
\end{equation}
The coordinate kets \(|r\rangle,|s\rangle\) lie in \(\mathbb C^8\),
whereas \(|T_r\rangle\!\rangle\in V\otimes V^*\) are nine-dimensional
matrix kets. Thus \(R(U)\in\mathbb R^{8\times8}\), and
\(|R(U)\rangle\!\rangle\in\mathbb C^8\otimes\mathbb C^8\).
The coefficients in Eq.~\eqref{w4:eq:adjoint-matrix} equal
\(\operatorname{tr}(T_r^\dagger UT_sU^\dagger)\), and this representation
satisfies \(R(U)^\top R(U)=I_8\) and \(R(U^\dagger)=R(U)^\top\).

The construction acts in the finite workspace
\begin{equation}\label{w4:eq:workspace}
\mathcal H=\mathcal H_{\mathsf S}\otimes\mathcal H_{\mathsf A},
\qquad(\dim\mathcal H_{\mathsf S},\dim\mathcal H_{\mathsf A})=(3,90),
\end{equation}
where \(\mathsf S\) is the {system} qutrit and \(\mathsf A\) is the complete
processing memory. Each query applies \(U_{\mathsf S}\otimes I_{\mathsf A}\).
Between queries, unitary gates independent of \(U\) act jointly on
\(\mathsf S\mathsf A\) to encode and decode the intermediate states.
Upright letters in ket or operator subscripts label systems or tensor
factors; numeric subscripts on \(I\) specify matrix sizes. The carrier
of an unsubscripted identity is determined by the composition.

For positive integers \(d_1,d_2\), define the fixed exchange
\(\mathrm{SWAP}_{d_1,d_2}:\mathbb C^{d_1}\otimes\mathbb C^{d_2}
\to\mathbb C^{d_2}\otimes\mathbb C^{d_1}\) by
\begin{equation}\label{w4:eq:swap}
\mathrm{SWAP}_{d_1,d_2}(|v\rangle\otimes|w\rangle)
=|w\rangle\otimes|v\rangle,
\qquad |v\rangle\in\mathbb C^{d_1},\quad |w\rangle\in\mathbb C^{d_2}.
\end{equation}
We write \(\mathrm{SWAP}=\mathrm{SWAP}_{3,3}\) for two qutrit factors;
in general \(\mathrm{SWAP}_{d_1,d_2}^\dagger=
\mathrm{SWAP}_{d_2,d_1}\).
Let \(\operatorname{Sym}^2V\) and
\(\operatorname{Sym}^2(V^*)\) be the six-dimensional symmetric
tensor subspaces of \(V^{\otimes2}\) and \((V^*)^{\otimes2}\), respectively.
Let \(A_6^{3,3}:\operatorname{Sym}^2V\hookrightarrow V\otimes V\)
be the inclusion in the fixed orthonormal symmetric basis. Its range
projector and the induced action of \(M\in\mathbb C^{3\times3}\) are
\[
\Pi_{\operatorname{Sym}^2}
=A_6^{3,3}\bigl(A_6^{3,3}\bigr)^{\dagger}
=\frac{I_9+\mathrm{SWAP}}2,\qquad
\operatorname{Sym}^2M
=\bigl(A_6^{3,3}\bigr)^{\dagger}(M\otimes M)A_6^{3,3}.
\]
The conjugate inclusion \(A_{\bar6}^{\bar3,\bar3}:=\bigl(A_6^{3,3}\bigr)^*\) and its range projector have the same real coefficients on
\(\operatorname{Sym}^2(V^*)\subset V^*\otimes V^*\).

We use the real Levi--Civita isometry
\(A_{\bar3}^{3,3}:V^{*}\to V^{\otimes2}\), defined by
\[
A_{\bar3}^{3,3}=\bigl[(|12\rangle-|21\rangle)\langle0|
+(|20\rangle-|02\rangle)\langle1|
+(|01\rangle-|10\rangle)\langle2|\bigr]/\sqrt2,
\]
which satisfies
\begin{equation}\label{w4:eq:iota}
\bigl(A_{\bar3}^{3,3}\bigr)^{\dagger}A_{\bar3}^{3,3}=I,\qquad
\mathrm{SWAP}A_{\bar3}^{3,3}=-A_{\bar3}^{3,3},\qquad
(U\otimes U)A_{\bar3}^{3,3}=A_{\bar3}^{3,3} U^*.
\end{equation}
In formulas for data maps, we apply this fixed coefficient matrix to the computational coordinates of \(V\).
We choose the \(18+18+9=45\)-dimensional intermediate space
\begin{equation}\label{w4:eq:middle-space}
\mathcal M_{4}=(V\otimes\operatorname{Sym}^2V)
\oplus(\operatorname{Sym}^2(V^*)\otimes V^*)
\oplus(V\otimes V^*).
\end{equation}
All three summands use fixed computational bases; the starred carriers have conjugate representation actions.
We store the midpoint's 45 intrinsic coordinates in the 270-dimensional workspace \(\mathcal H\) through a fixed isometry \(A_{\mathcal M_4}^{\mathcal H}:\mathcal M_{4}\to\mathcal H\), with \(\bigl(A_{\mathcal M_4}^{\mathcal H}\bigr)^\dagger A_{\mathcal M_4}^{\mathcal H}=I_{\mathcal M_{4}}\).
This \(270\times45\) embedding fixes the relation between representation order and the physical gate coordinates used below and in the accompanying notebook.
{The main text abbreviates $\mathcal M_4$ and $A_{\mathcal M_4}^{\mathcal H}$
as $\mathcal M$ and $A_{\mathcal M}^{\mathcal H}$, respectively.}
We later complete these isometries to unitaries on the workspace.
Define three maps from \(V\) into the first, second, and third
summands of \(\mathcal M_{4}\), respectively. Their formulas use the
fixed coefficient matrices of \(\bigl(A_6^{3,3}\bigr)^{\dagger}\), its conjugate, and
\(A_{\bar3}^{3,3}\) in the stated carrier bases:
\begin{equation}\label{w4:eq:phi-components}
\begin{aligned}
\Phi_{3,6}(U)&=\frac1{\sqrt2}\Bigl(U\otimes\bigl(A_6^{3,3}\bigr)^{\dagger}\Bigr)
  (|I\rangle\!\rangle\otimes U^{\top}),\\
\Phi_{\bar6,\bar3}(U)&=\frac1{\sqrt2}\Bigl(\bigl(A_{\bar6}^{\bar3,\bar3}\bigr)^{\dagger}\otimes U^\dagger\Bigr)
  (I\otimes|I\rangle\!\rangle),\\
\Phi_{3,\bar3}(U)&=(I\otimes U^\dagger)A_{\bar3}^{3,3}.
\end{aligned}
\end{equation}
The contractions with these symmetric Clebsch--Gordan (CG) maps~\cite[App.~A.2]{grinko} give the
first two maps their \(18\)-dimensional output coordinates.
For \(0\le\alpha\le\tfrac12\), define their weighted direct sum
\(\Phi_\alpha(U):V\longrightarrow\mathcal M_{4}\) by
\begin{equation}\label{w4:eq:phi-family}
\Phi_\alpha(U)=\sqrt\alpha\,\Phi_{3,6}(U)
\oplus\sqrt\alpha\,\Phi_{\bar6,\bar3}(U)
\oplus\sqrt{1-2\alpha}\,\Phi_{3,\bar3}(U).
\end{equation}
{The main text uses the abbreviation $\Phi(U)=\Phi_{1/5}(U)$.}
Each component is an isometric embedding. The unitary oracle
factors reduce the isometry checks for the first two maps to \(U=I\).
Since \(\Phi_{\bar6,\bar3}(I)=\mathrm{SWAP}_{3,6}\Phi_{3,6}(I)\),
both have the same Gram matrix, and
\[
\Phi_{3,6}(I)^\dagger\Phi_{3,6}(I)
=\frac12(\langle\!\langle I|\otimes I)(I\otimes \Pi_{\operatorname{Sym}^2})
(|I\rangle\!\rangle\otimes I)=\frac{3I_3+I_3}{4}=I_3.
\]

For the third component, isometry follows directly from
Eq.~\eqref{w4:eq:iota}.
The squared weights in Eq.~\eqref{w4:eq:phi-family} sum to one, so
\(\Phi_\alpha(U):V\hookrightarrow\mathcal M_{4}\) is also an isometric
embedding and sends every normalized \(|\psi\rangle\) to a normalized
intermediate state. The physical realization of \(\Phi_\alpha(U)\) is
\(A_{\mathcal M_4}^{\mathcal H}\Phi_\alpha(U):V\hookrightarrow\mathcal H\).

Transpose covariance reduces the encoder construction to matching seeds at \(U=I\): two covariant families with the same output representation and seed agree for every \(U\).
In Sec.~\ref{w4:chap:covariance}, we derive this reconstruction principle and use representation multiplicities to count the allowed seed coefficients.
Fixed encodings and queries then determine how the output representation evolves.
These ingredients give the encoder for \(\Phi_{1/5}(U)\), the fixed flip, and the reversed decoder in Sec.~\ref{w4:chap:processing}, including their tensor embeddings, phases, and unitary completions.
The resulting circuit is verified by exact symbolic computation in Sec.~\ref{sm:symbolic}.

\subsection{Catalyst preparation}\label{w4:sec:catalyst-preparation}

Our preparation method produces the catalyst family
\begin{equation}\label{w4:eq:catalyst}
|c_a(U)\rangle_{\mathsf P}
=\sqrt{1-a}
\oplus\sqrt{\frac a8}|R(U)\rangle\!\rangle,
\qquad 0\le a\le1.
\end{equation}
We write \(|c_a(U)\rangle_{\mathsf A}\) for the image under the first-summand inclusion into the processing memory.
The singlet--adjoint decomposition of \(U\otimes U^*\) produces the two components of \(|c_a(U)\rangle_{\mathsf P}\), so we can prepare the catalyst before the data arrive using temporary auxiliary systems.

\begin{wlemma}[Catalyst-family preparation]\label{w4:lem:catalyst-preparation}
For every known \(a\in[0,1]\), the state \(|c_a(U)\rangle_{\mathsf P}\)
can be prepared deterministically using at most three forward queries
to \(U\in SU(3)\) and known operations independent of \(U\).
\end{wlemma}

\begin{wproof}
Prepending the column \(|I/\sqrt3\rangle\!\rangle\) to the rightmost matrix in Eq.~\eqref{w4:eq:adjoint-matrix} gives a unitary \(G_{3\bar3}\in\mathrm U(9)\) satisfying
\begin{equation}\label{w4:eq:preparation-decomposition}
G_{3\bar3}^\dagger(U\otimes U^*)G_{3\bar3}=1\oplus R(U).
\end{equation}
We first prepare the fixed input state from Eq.~\eqref{w4:eq:catalyst},
\begin{equation}\label{w4:eq:preparation-seed}
|c_a(I)\rangle_{\mathsf P}
=\sqrt{1-a}
\oplus\sqrt{\frac a8}|I_8\rangle\!\rangle
\end{equation}
using operations independent of \(U\); here \(|I_8\rangle\!\rangle/\sqrt8\) is the normalized Bell state on \(\mathbb C^8\otimes\mathbb C^8\).
{We} apply the group action to the representation factor and preserve its reference{, using} the representation-space {framework} of Ref.~\cite{hayashi2026}:
\begin{equation}\label{w4:eq:preparation-action}
\bigl[1\oplus(R(U)\otimes I_8)\bigr]|c_a(I)\rangle_{\mathsf P}
=|c_a(U)\rangle_{\mathsf P}.
\end{equation}
Tensoring Eq.~\eqref{w4:eq:preparation-decomposition} with the identity on an
eight-dimensional reference gives \((1\oplus R(U))\otimes I_8\).
A fixed basis permutation identifies this \(9\times8=72\)-dimensional
space with \(\mathcal H_{\mathsf P}\oplus\mathbb C^7\), where the action is
\[
\bigl[1\oplus(R(U)\otimes I_8)\bigr]\oplus I_7.
\]
Fixed gates and \(U\otimes U^*\) therefore implement Eq.~\eqref{w4:eq:preparation-action}.
The Levi--Civita construction in Eq.~\eqref{w4:eq:iota} realizes \(U^*\) with two forward queries on two encoded qutrits~\cite{miyazaki}.
One further query supplies \(U\), completing the three-query preparation of \(|c_a(U)\rangle_{\mathsf P}\).
For \(a=0\), the known state in the trivial summand requires no query.
\end{wproof}

Taking \(a=4/5\) supplies the catalyst for Theorem~\ref{w4:thm:S1}.
Once it is prepared, processing takes place entirely in the workspace
of Eq.~\eqref{w4:eq:workspace}.

\section{Transpose-covariant maps}\label{w4:chap:covariance}

\phantomsection\label{w4:tc:representation-labels}
\noindent\textbf{Notation.} We use dimensions as shorthand for the irreducible representations appearing here, with a bar denoting the conjugate representation~\cite{slansky1981}.
Thus \(3\) denotes the fundamental representation \(g\mapsto g\), \(\bar3\) its complex conjugate \(g\mapsto g^*\), and \(6\) and \(\bar6\) their respective symmetric squares.
The label \(1\) denotes the trivial representation, while \(8\) denotes the adjoint representation on traceless \(3\times3\) matrices, with action \(X\mapsto gXg^\dagger\).
Its matrix in the basis \(T_0,\ldots,T_7\) is \(R(g)\) from Eq.~\eqref{w4:eq:adjoint-matrix}.

Transpose covariance lets us design one circuit for every unknown \(U\) by relating the map at different parameters.
We apply equivariance, Schur multiplicity, and tensor--Hom duality~\cite{etingof2011} directly to the coherent linear maps of the construction.
At the level of Choi operators and superchannels, Chiribella et al. describe the seed--stabilizer structure of covariant quantum networks~\cite{covariant2009}, and Quintino and Ebler derive two-sided symmetries and reconstruction from the identity for anti-homomorphic unitary transformations~\cite{quintinoebler}.
We formulate this covariance directly for linear maps, retaining amplitudes and phases. The starting point is the identity
\begin{equation}
  (gUh^{\top})^{\top}g^*=hU^{\top}.
  \label{w4:tc:transpose-motivation}
\end{equation}
The input action \(g^*\) cancels \(g^\top\) and leaves the output action \(h\).
We extend this relation to larger output spaces and, by analogy with covariant state estimation~\cite{chiribella2004}, call the value at \(U=I\) the \emph{seed map}, which determines the full family.
We can therefore match seeds while fixed encodings and queries change the output representation.

\subsection{Transpose-covariant maps}\label{w4:tc:sec:definition}

\begin{wdefinition}[Transpose covariance]\label{w4:def:tc}
Let \(V=\mathbb C^3\), and let the output space \(\mathcal X\) carry a unitary representation \(\rho\) of \(SU(3)\times SU(3)\). A family of linear maps \(K(U):V\to\mathcal X\) satisfying
\begin{equation}
  \boxed{K(gUh^{\top})g^*=\rho(g,h)K(U)}
  \qquad\text{for all }g,h,U\in SU(3),
  \label{w4:tc:definition}
\end{equation}
is called \emph{transpose covariant} with respect to \(\rho\).
\end{wdefinition}

The fixed input action \(g^*\) links the families across circuit stages, while the output representation identifies the space occupied at each stage.
The intermediate maps of Theorem~\ref{w4:thm:S1} obey this covariance relation.

\begin{wlemma}[Transpose covariance of the intermediate maps]\label{w4:lem:tc-phi}
The maps \(\Phi_{3,6},\Phi_{\bar6,\bar3},\Phi_{3,\bar3}\) defined in Eq.~\eqref{w4:eq:phi-components} are transpose covariant with respect to
\begin{equation}
  \rho_1(g,h)=g\otimes\operatorname{Sym}^2h,\qquad
  \rho_2(g,h)=\operatorname{Sym}^2(g^*)\otimes h^*,\qquad
  \rho_3(g,h)=g\otimes h^*
\label{w4:tc:phi-actions}
\end{equation}
respectively.
Their weighted direct sum \(\Phi_\alpha\), defined in
Eq.~\eqref{w4:eq:phi-family}, is transpose covariant with respect to
\(\rho_1\oplus\rho_2\oplus\rho_3\), for every \(0\le\alpha\le1/2\).
\end{wlemma}

\begin{wproof}
We prove the covariance relations by moving the group actions through the Bell vector and the symmetric and antisymmetric embeddings.
For \(M\in\mathbb C^{3\times3}\), we use the transpose trick~\cite{bennett1996,werner2001} with the unnormalized Bell vector, together with symmetric projection and compression and the \(d=3\) antisymmetric identity of Eq.~\eqref{w4:eq:iota}~\cite{miyazaki}:
\begin{equation}
\begin{aligned}
  (M\otimes I)|I\rangle\!\rangle
    &=(I\otimes M^{\top})|I\rangle\!\rangle,\\
  \Pi_{\operatorname{Sym}^2}(M\otimes M)&=(M\otimes M)\Pi_{\operatorname{Sym}^2},\\
  \bigl(A_6^{3,3}\bigr)^{\dagger}(M\otimes M)
    &=\operatorname{Sym}^2M\,\bigl(A_6^{3,3}\bigr)^{\dagger},\\
  A_{\bar3}^{3,3} g^*&=(g\otimes g)A_{\bar3}^{3,3}.
\end{aligned}
\label{w4:tc:basic-identities}
\end{equation}
The transpose trick and symmetric compression give
\begin{align}
  \Phi_{3,6}(gUh^{\top})g^*
  &=\frac1{\sqrt2}\Bigl(gUh^{\top}\otimes \bigl(A_6^{3,3}\bigr)^{\dagger}\Bigr)
       \bigl(|I\rangle\!\rangle\otimes hU^{\top}\bigr)\notag\\
  &=\frac1{\sqrt2}\Bigl(gU\otimes \bigl(A_6^{3,3}\bigr)^{\dagger}\Bigr)
       (I\otimes h\otimes h)
       \bigl(|I\rangle\!\rangle\otimes U^{\top}\bigr)\notag\\
  &=\rho_1(g,h)\Phi_{3,6}(U),\qquad\text{and}
  \label{w4:tc:phi-one-proof}
\end{align}
\begin{align}
  \Phi_{\bar6,\bar3}(gUh^\top)g^*
  &=\frac1{\sqrt2}\Bigl(\bigl(A_{\bar6}^{\bar3,\bar3}\bigr)^{\dagger}\otimes h^*U^\dagger g^\dagger\Bigr)
       (g^*\otimes|I\rangle\!\rangle)\notag\\
  &=\frac1{\sqrt2}\Bigl(\bigl(A_{\bar6}^{\bar3,\bar3}\bigr)^{\dagger}\otimes h^*U^\dagger\Bigr)
       (g^*\otimes g^*\otimes I)
       (I\otimes|I\rangle\!\rangle)\notag\\
  &=\rho_2(g,h)\Phi_{\bar6,\bar3}(U).
  \label{w4:tc:phi-two-proof}
\end{align}
The Levi--Civita identity gives
\begin{equation}
  \Phi_{3,\bar3}(gUh^\top)g^*=(g\otimes h^*U^\dagger)A_{\bar3}^{3,3}=\rho_3(g,h)\Phi_{3,\bar3}(U).
  \label{w4:tc:phi-three-proof}
\end{equation}
Taking the same fixed weighted direct sum of these identities gives the covariance of \(\Phi_\alpha\).
\end{wproof}

\subsection{Seed reconstruction and normalization}\label{w4:tc:sec:seed}

A transpose-covariant family is determined by its value at \(I\).
Indeed, the pair \((g,h)=(U,I)\) takes \(I\) to any \(U\), so covariance reconstructs \(K(U)\) from \(K(I)\).
The seed must also be consistent with the pairs \((g,g^*)\), which leave \(I\) unchanged because \(g(g^*)^\top=I\).
These pairs therefore impose an intertwining constraint on \(K(I)\).
The reconstruction and constraint below express the seed--stabilizer principle for covariant instruments~\cite[Theorem~2]{covariant2009} and identity-based reconstruction for anti-homomorphic unitary transformations at the channel level~\cite[Lemma~1]{quintinoebler} in terms of coherent maps.

\begin{wproposition}[Seed reconstruction and constraint]\label{w4:lem:tc-seed}
If \(K\) is transpose covariant with respect to a unitary representation \(\rho\), then it satisfies the \emph{seed reconstruction identity}
\begin{equation}
  \boxed{K(U)=\rho(U,I)K(I)U^{\top},}
  \label{w4:tc:seed-reconstruction}
\end{equation}
and the \emph{seed constraint}
\begin{equation}
  \boxed{K(I)g^*=\rho(g,g^*)K(I).}
  \label{w4:tc:seed-constraint}
\end{equation}
Consequently, two covariant families with the same output representation agree for every \(U\) whenever their seeds at \(I\) agree.
\end{wproposition}

\begin{wproof}
Both identities follow from covariance at \(I\). Setting \(U=I\) in Eq.~\eqref{w4:tc:definition} gives
\begin{equation}
  K(gh^{\top})g^*=\rho(g,h)K(I).
  \label{w4:tc:at-identity}
\end{equation}
The choices \(h=I\) and \(h=g^*\) respectively reconstruct the family and constrain its seed. For \(h=I\), Eq.~\eqref{w4:tc:at-identity} gives
\begin{equation}
  K(g)g^*=\rho(g,I)K(I),
  \qquad
  K(g)=\rho(g,I)K(I)g^{\top}.
  \label{w4:tc:reconstruction-step}
\end{equation}
Renaming \(g\) as \(U\) yields the reconstruction identity~\eqref{w4:tc:seed-reconstruction}. The uniqueness statement follows directly from Eq.~\eqref{w4:tc:seed-reconstruction}.
For the second choice, \(g(g^*)^{\top}=I\) gives
\[
  K(I)g^*=\rho(g,g^*)K(I),
\]
which is the seed constraint~\eqref{w4:tc:seed-constraint}.
\end{wproof}

The seed constraint also fixes normalization up to a scalar, because the input action is irreducible.

\begin{wcorollary}[Uniform normalization]\label{w4:lem:tc-normalization}
If \(K\) is transpose covariant with respect to a unitary representation \(\rho\), then there is a constant \(\gamma\ge0\), independent of \(U\), such that
\begin{equation}
  K(U)^\dagger K(U)=\gamma I_V.
  \label{w4:tc:uniform-normalization}
\end{equation}
Thus \(\|K(I)|\psi\rangle\|^2=\gamma\) for any normalized input state \(|\psi\rangle\). Evaluating this one output norm at \(I\) determines whether the entire family consists of isometric embeddings.
\end{wcorollary}

\begin{wproof}
Using the seed constraint~\eqref{w4:tc:seed-constraint} and the unitarity of \(\rho(g,g^*)\), we obtain
\begin{align}
  g^{\top}K(I)^\dagger K(I)g^*
  &=\bigl[\rho(g,g^*)K(I)\bigr]^\dagger
    \bigl[\rho(g,g^*)K(I)\bigr]\notag\\
  &=K(I)^\dagger\rho(g,g^*)^\dagger\rho(g,g^*)K(I)=K(I)^\dagger K(I),
  \label{w4:tc:normalization-adjoint-product}
\end{align}
which implies that \(K(I)^\dagger K(I)\) commutes with every \(g^*\).
Since \(g\mapsto g^*\) is irreducible, Schur's lemma~\cite{fultonharris1991} in its scalar-intertwiner form~\cite[Cor.~2.3.10]{etingof2011} gives \(K(I)^\dagger K(I)=\gamma I_V\), with \(\gamma\ge0\) by positivity.
The reconstruction identity~\eqref{w4:tc:seed-reconstruction} then yields
\begin{equation}
  K(U)^\dagger K(U)
  =U^*K(I)^\dagger K(I)U^{\top}
  =\gamma I_V,
  \label{w4:tc:normalization-all-u}
\end{equation}
proving Eq.~\eqref{w4:tc:uniform-normalization}.
\end{wproof}

\subsection{Representation blocks and seed coefficients}\label{w4:tc:sec:blocks}

Representation multiplicities count the independent seed coefficients.
We decompose the output into blocks on which \(g\) and \(h\) act irreducibly on separate factors, then impose the seed constraint within each block.
Schur's lemma identifies the allowed components and leaves one coefficient for each copy of the input representation.
The underlying compact-group tensor representation decomposes into irreducible and multiplicity spaces~\cite{grinko}.

\begin{wdefinition}[Left and right representation blocks]\label{w4:def:tc-block}
Let \(\lambda\) and \(\mu\) be irreducible unitary representations of \(SU(3)\) on \(W_\lambda\) and \(W_\mu\), respectively.

A subspace of the output space is called a \((\lambda,\mu)\) block if it is the image of a fixed isometric embedding
\begin{equation}
  A_{\lambda,\mu}^{\mathcal X}:W_\lambda\otimes W_\mu\longrightarrow\mathcal X,
  \qquad \bigl(A_{\lambda,\mu}^{\mathcal X}\bigr)^\dagger A_{\lambda,\mu}^{\mathcal X}=I,
  \label{w4:tc:block-embedding}
\end{equation}
satisfying, for all \(g,h\),
\begin{equation}
  \boxed{\rho(g,h)A_{\lambda,\mu}^{\mathcal X}=A_{\lambda,\mu}^{\mathcal X}\bigl[\lambda(g)\otimes\mu(h)\bigr].}
  \label{w4:tc:block-intertwining}
\end{equation}
\end{wdefinition}

We use the letter \(A\) for chosen fixed isometric embeddings, with CG maps as a special case.

Equation~\eqref{w4:tc:block-intertwining} ensures that \(\rho(g,h)\) preserves the image of \(A_{\lambda,\mu}^{\mathcal X}\)~\cite[Def.~2.3.6]{etingof2011}.
In the coordinates set by this embedding, the restricted action is \(\lambda(g)\otimes\mu(h)\).

A block decomposition is an orthogonal direct sum of these subspaces.
Write \(d_\lambda=\dim W_\lambda\) and \(d_\mu=\dim W_\mu\).
The labels record the action types of the independent variables
\(g\) and \(h\), and a \((\lambda,\mu)\) block has dimension
\(d_\lambda d_\mu\). The same type may occur several times;
choosing \(A_{\lambda,\mu}^{\mathcal X}\) specifies one copy.
When several copies need to be distinguished, we write \(A_{\lambda,\mu;r}^{\mathcal X}\), with copy index \(r\).

The intermediate space gives three concrete examples of these blocks.
\begin{wexample}[Intermediate representation blocks]\label{w4:tc:remark:midpoint-blocks}
The three summands of \(\mathcal M_4\) in Eq.~\eqref{w4:eq:middle-space}, carrying \(\rho_1(g,h)=g\otimes\operatorname{Sym}^2h\), \(\rho_2(g,h)=\operatorname{Sym}^2(g^*)\otimes h^*\), and \(\rho_3(g,h)=g\otimes h^*\) from Eq.~\eqref{w4:tc:phi-actions}, have types \((3,6)\), \((\bar6,\bar3)\), and \((3,\bar3)\), respectively.
\end{wexample}

The input action in the covariance relation~\eqref{w4:tc:definition} is \(g^*\), so it carries \(\bar3\). Restricting a \((\lambda,\mu)\) block to \((g,h)=(g,g^*)\), the seed constraint~\eqref{w4:tc:seed-constraint} makes its seed an intertwiner from \(\bar3\) to \(g\mapsto\lambda(g)\otimes\mu(g^*)\).

\begin{wproposition}[Multiplicity of block seeds]\label{w4:lem:tc-multiplicity}
Let \(K\) be transpose covariant with respect to \(\rho\), and let \(A_{\lambda,\mu}^{\mathcal X}\) specify one of its \((\lambda,\mu)\) blocks. The seed component \(S=\bigl(A_{\lambda,\mu}^{\mathcal X}\bigr)^\dagger K(I)\) in that block satisfies
\begin{equation}
  Sg^*=\bigl[\lambda(g)\otimes\mu(g^*)\bigr]S.
  \label{w4:tc:block-seed-constraint}
\end{equation}
The space of such seeds has complex dimension \(m\), which is the multiplicity of \(\bar3\) in \(g\mapsto\lambda(g)\otimes\mu(g^*)\). Thus \(m=0\) forces the branch to vanish, whereas \(m=1\) determines its seed up to a complex scalar multiplying a fixed isometric embedding.
\end{wproposition}

\begin{wproof}
Because both representations are unitary, the block relation~\eqref{w4:tc:block-intertwining} is equivalent to
\begin{equation}
  \bigl(A_{\lambda,\mu}^{\mathcal X}\bigr)^\dagger\rho(g,h)
  =\bigl[\lambda(g)\otimes\mu(h)\bigr]\bigl(A_{\lambda,\mu}^{\mathcal X}\bigr)^\dagger.
  \label{w4:tc:block-extraction}
\end{equation}
Applying this identity to the seed constraint~\eqref{w4:tc:seed-constraint} with \(h=g^*\) gives Eq.~\eqref{w4:tc:block-seed-constraint} for \(S=\bigl(A_{\lambda,\mu}^{\mathcal X}\bigr)^\dagger K(I)\).

Decompose \(g\mapsto\lambda(g)\otimes\mu(g^*)\) into orthogonal irreducible summands. By Schur's lemma, each component of \(S\) vanishes unless its output representation is equivalent to the input \(\bar3\). On each equivalent copy, the intertwiner is unique up to a complex scalar.

For the \(m\) copies of \(\bar3\), choose fixed isometric intertwiners \(S_{\lambda,\mu}^{(j)}:V\to W_\lambda\otimes W_\mu\), each with image in its corresponding summand. Every allowed seed then has the form
\begin{equation}
  S=\sum_{j=1}^m b_jS_{\lambda,\mu}^{(j)},\qquad b_j\in\mathbb C.
  \label{w4:tc:seed-multiplicity-expansion}
\end{equation}
Conversely, every such sum satisfies the block constraint. The summands have orthogonal images, so these \(m\) coefficients are independent. For \(m=0\), the seed vanishes and reconstruction makes the entire branch vanish; for \(m=1\), one scalar remains.
\end{wproof}

The following example applies the multiplicity rule to the antisymmetric seed used in the construction.
\begin{wexample}[Antisymmetric seed]
Recall from Eq.~\eqref{w4:tc:phi-actions} that the third intermediate block has action \(\rho_3(g,h)=g\otimes h^*\) and type \((3,\bar3)\). Setting \(h=g^*\) gives
\begin{equation}
  Sg^*=(g\otimes g)S.
  \label{w4:tc:antisymmetric-seed}
\end{equation}
The output splits into its six-dimensional symmetric and three-dimensional antisymmetric parts, with representation types \(6\oplus\bar3\). The input maps only into the antisymmetric part, which occurs once. Hence \(S=bA_{\bar3}^{3,3}\), with \(b\in\mathbb C\). The \(d=3\) antisymmetric encoding of Miyazaki et al.~\cite{miyazaki} gives this fixed seed, which is \(\Phi_{3,\bar3}(I)\) when \(b=1\). Equation~\eqref{w4:eq:iota} gives \(\mathrm{SWAP}A_{\bar3}^{3,3}=-A_{\bar3}^{3,3}\), so the negative exchange used in the midpoint flip preserves its seed.
\end{wexample}

\subsection{Tensor--Hom duality and seed multiplicities}\label{w4:sec:tc-reshaping}

Tensor--Hom duality~\cite[Problems~2.11.3(c), 2.14.3]{etingof2011} moves the second output factor of a block seed to its input, reducing seed counting to the decomposition of \(\bar3\otimes\mu\).
We realize this correspondence by contraction with the Bell vector in a basis satisfying
\begin{equation}
  \mu(g^*)=\mu(g)^*.
  \label{w4:tc:compatible-basis}
\end{equation}
This choice supplements the arbitrary fixed-basis convention of Proposition~\ref{w4:lem:tc-multiplicity}.

With Eq.~\eqref{w4:tc:compatible-basis} substituted into the block constraint~\eqref{w4:tc:block-seed-constraint}, the block seed satisfies
\begin{equation}
  Sg^*=\bigl[\lambda(g)\otimes\mu(g)^*\bigr]S,
  \qquad S:V\to W_\lambda\otimes W_\mu.
  \label{w4:tc:compatible-seed}
\end{equation}
Let \(|I_{d_\mu}\rangle\!\rangle\) denote the unnormalized Bell vector in the compatible basis, and define the linear map
\begin{equation}
  B=(I_{W_\lambda}\otimes\langle\!\langle I_{d_\mu}|)
       (S\otimes I_{W_\mu})
  :V\otimes W_\mu\to W_\lambda.
  \label{w4:tc:bell-reshaping}
\end{equation}
This reshaping is complex linear, with no coefficient conjugation; its inverse is
\begin{equation}
  S=(B\otimes I_{W_\mu})
      (I_V\otimes|I_{d_\mu}\rangle\!\rangle).
  \label{w4:tc:inverse-reshaping}
\end{equation}
Unitarity makes the Bell pairing invariant under \(\mu(g)^*\otimes\mu(g)\), so Eq.~\eqref{w4:tc:compatible-seed} is equivalent to
\begin{align}
  B\bigl(g^*\otimes\mu(g)\bigr)
  &=\lambda(g)B.
  \label{w4:tc:reshaped-intertwining}
\end{align}
Thus the number of allowed seed coefficients equals the multiplicity of \(\lambda\) in \(\bar3\otimes\mu\).

This correspondence recovers the seed count of Proposition~\ref{w4:lem:tc-multiplicity}.
This is the reciprocal multiplicity identity~\cite[Eq.~(A29)]{grinko}.
The reshaping above preserves the original coefficients; the normalized dual-Clebsch--Gordan convention~\cite[Eq.~(A31)]{grinko} instead uses coefficient conjugation and a dimension factor.
{For a transpose-covariant family with output action
\(g\otimes\lambda(g)\otimes\mu(h)\), the left representation types lie in
\(3\otimes\lambda\). In a right basis satisfying
Eq.~\eqref{w4:tc:compatible-basis}, the seed in a
\((\lambda_{\mathrm{out}},\mu)\) block can be nonzero only if
\(\lambda_{\mathrm{out}}\) also occurs in \(\bar3\otimes\mu\). Hence}
\begin{equation}\label{w4:eq:query-selection-rule}
\lambda_{\mathrm{out}}\in
\operatorname{Irr}(3\otimes\lambda)
\cap\operatorname{Irr}(\bar3\otimes\mu).
\end{equation}
Here \(\operatorname{Irr}\) records representation types{, and }\(\mu\) {is the }right output label.
The {selection rule} can exclude a branch created by a query.
\begin{wexample}[Excluding the \(10\) branch]
Let \(10\) denote the ten-dimensional symmetric cube
\(\operatorname{Sym}^3V\), and \(\overline{10}\) its conjugate representation.
{For \(\lambda=6\) and \(\mu=\bar6\), the decompositions}
\(3\otimes6=8\oplus10\) {and}
\(\bar3\otimes\bar6=8\oplus\overline{10}\) {show that Eq.~\eqref{w4:eq:query-selection-rule} allows} only the left \(8\) {to carry} a seed.
The \(10\) seed vanishes, and reconstruction makes that branch vanish for every \(U\).

This situation arises from the initial block \((6,8)\). The fixed encoding embeds its right \(8\) into \(\bar3\otimes\bar6\), leaving right memory \(\bar6\); the query then supplies the left \(3\) in \(3\otimes6\). Sec.~\ref{w4:ex:first-6-8} gives the corresponding circuit calculation.
\end{wexample}

\subsection{Fixed encodings and query updates}\label{w4:tc:sec:queries}

Fixed encodings expose the qutrit on which each query acts.
Within a representation block, the encoding makes this qutrit a tensor factor with the required group action.
Composition of intertwining maps preserves covariance; the query then changes the action on that factor.

\begin{wlemma}\label{w4:lem:tc-encoding}
Let \(K(U):V\to\mathcal X\) be transpose covariant with respect to
\(\rho\). Suppose that a fixed linear encoding
\(\mathcal E_{\mathrm{fix}}:\mathcal X\to\widetilde{\mathcal X}\), independent of \(U\),
maps into a space carrying \(\widetilde\rho\) and satisfies

  \begin{equation}
    \mathcal E_{\mathrm{fix}}\rho(g,h)=\widetilde\rho(g,h)\mathcal E_{\mathrm{fix}}.
    \label{w4:tc:fixed-encoding-intertwining}
  \end{equation}
Then the family \(U\mapsto \mathcal E_{\mathrm{fix}}K(U)\) is transpose covariant with respect to \(\widetilde\rho\).

\end{wlemma}

\begin{wproof}
The fixed encoding carries the covariance of \(K\) to the new output representation.
Because \(\mathcal E_{\mathrm{fix}}\) is independent of \(U\), the transformed map is \(\mathcal E_{\mathrm{fix}}K(gUh^{\top})\).
Using the intertwining relation~\eqref{w4:tc:fixed-encoding-intertwining}, we obtain
\begin{equation}
  \mathcal E_{\mathrm{fix}}K(gUh^\top)g^*=\mathcal E_{\mathrm{fix}}\rho(g,h)K(U)=\widetilde\rho(g,h)\mathcal E_{\mathrm{fix}}K(U).
  \label{w4:tc:fixed-encoding-proof}
\end{equation}
By definition, \(\mathcal E_{\mathrm{fix}}K\) is transpose covariant with respect to \(\widetilde\rho\).
\end{wproof}

\begin{wlemma}\label{w4:lem:tc-query}
Suppose \(K(U):V\to V\otimes\mathcal M\), where \(\mathcal M\)
is the memory and the first output factor is the query qutrit.
If \(K\) is transpose covariant with respect to
\(h^*\otimes\rho_{\mathcal M}(g,h)\), where \(\rho_{\mathcal M}\) is a unitary
representation on \(\mathcal M\), then the family
\(U\mapsto(U\otimes I_{\mathcal M})K(U)\) after one forward query
is transpose covariant with respect to \(g\otimes\rho_{\mathcal M}(g,h)\).
\end{wlemma}

\begin{wproof}
A query changes the qutrit action while preserving the memory action.
Before the query, covariance reads
\begin{equation}
  K(gUh^{\top})g^*
  =\bigl[h^*\otimes\rho_{\mathcal M}(g,h)\bigr]K(U).
  \label{w4:tc:before-query-covariance}
\end{equation}
After the query, the map is \((U\otimes I)K(U)\).
Replacing \(U\) by \(gUh^{\top}\) in both factors gives \(\bigl((gUh^{\top})\otimes I\bigr)K(gUh^{\top})\).
Equation~\eqref{w4:tc:before-query-covariance} then yields
\begin{align}
  \bigl((gUh^{\top})\otimes I\bigr)K(gUh^{\top})g^*
  &=\bigl((gUh^{\top})\otimes I\bigr)
      \bigl[h^*\otimes\rho_{\mathcal M}(g,h)\bigr]K(U)\notag\\
  &=\bigl[(gUh^{\top})h^*\otimes\rho_{\mathcal M}(g,h)\bigr]K(U)\notag\\
  &=\bigl[gU\otimes\rho_{\mathcal M}(g,h)\bigr]K(U)\notag\\
  &=\bigl[g\otimes\rho_{\mathcal M}(g,h)\bigr](U\otimes I)K(U).
  \label{w4:tc:query-update-proof}
\end{align}
The first equality uses Eq.~\eqref{w4:tc:before-query-covariance}, and the last is the covariance relation for \(U\mapsto(U\otimes I)K(U)\) with output action \(g\otimes\rho_{\mathcal M}(g,h)\).
\end{wproof}

Alternating fixed encodings with forward queries traces a representation path.
Each query replaces \(h^*\) on the query qutrit by \(g\), while preserving the memory action \(\rho_{\mathcal M}(g,h)\) and input action \(g^*\).
Proposition~\ref{w4:lem:tc-seed} then lets us compare the families along that path through their seeds.
Grinko et al.~\cite[Theorems~3 and~4]{grinko} realize such paths at the network level with Clebsch--Gordan and dual-Clebsch--Gordan transforms and isometries on multiplicity spaces.
Here the coherent single-query identity gives the update used by our circuit.

Isometric encodings on the required subspaces yield physical gates when they fit in one common workspace.
Completing their orthonormal bases extends each encoding to a unitary.
Our two-query encoder in Sec.~\ref{w4:chap:processing} meets these conditions.

\subsection{A fixed-flip criterion}\label{w4:tc:sec:flip}

The following criterion implements the desired \(U^{\top}\) action without another query: a fixed unitary exchanges the representation actions and fixes the seed.

\begin{wproposition}[Fixed flip criterion]\label{w4:lem:tc-flip}
Let \(K(U):V\to\mathcal X\) be transpose covariant with respect to
\(\rho\). Suppose that a unitary
\(J:\mathcal X\to\mathcal X\), independent of \(U\),
satisfies the representation exchange condition
\begin{equation}
  J\rho(g,h)=\rho(h^*,g^*)J,
  \label{w4:tc:flip-representation}
\end{equation}
and the seed invariance condition
\begin{equation}
  JK(I)=K(I).
  \label{w4:tc:flip-seed}
\end{equation}
Then
\begin{equation}
  \boxed{JK(U)=K(U^\dagger)U^{\top}.}
  \label{w4:tc:flip-family}
\end{equation}
\end{wproposition}

\begin{wproof}
Seed reconstruction~\eqref{w4:tc:seed-reconstruction} carries the representation exchange~\eqref{w4:tc:flip-representation} and seed invariance~\eqref{w4:tc:flip-seed} to the whole family:
\begin{equation}
  JK(U)=J\rho(U,I)K(I)U^{\top}=\rho(I,U^*)JK(I)U^{\top}=\rho(I,U^*)K(I)U^{\top}=K(U^\dagger)U^{\top}.
  \label{w4:tc:flip-proof}
\end{equation}
For the last equality, first set the original parameter \(U=I\) and \(g=I\) in the covariance relation~\eqref{w4:tc:definition}, obtaining
\[
  K(h^{\top})=\rho(I,h)K(I),
\]
then set \(h=U^*\), which gives
\begin{equation}
  K(U^\dagger)=\rho(I,U^*)K(I).
  \label{w4:tc:adjoint-parameter}
\end{equation}
This proves Eq.~\eqref{w4:tc:flip-family}.
\end{wproof}

\section{Construction and verification}
\label{w4:chap:processing}

We construct a two-query representation path with the output action and seed of \(\Phi_{1/5}\).
Fixed encodings and queries preserve covariance by Lemmas~\ref{w4:lem:tc-encoding} and~\ref{w4:lem:tc-query}, so Proposition~\ref{w4:lem:tc-seed} identifies the output family from its seed.
The three fixed gates act first on prescribed subspaces of the \(270\)-dimensional workspace; we complete their orthonormal bases to obtain physical unitaries.
Representation decompositions select the occupied branches, while tensor embeddings fix their amplitudes and phases.
The final two subsections implement the fixed flip and reverse decoding, completing the four-query processor and restoring the catalyst.
The full circuit is also verified symbolically in Sec.~\ref{sm:mathematica}.

\subsection{Initial representation blocks}
\label{w4:sec:branches}

We realize the adjoint representation on the traceless matrix space
\[
 W_8=\{X\in M_3(\mathbb C):\operatorname{tr}X=0\},
\qquad \langle X,Y\rangle=\operatorname{tr}(X^\dagger Y),
\]
with adjoint action
\begin{equation}\label{w4:eq:adjoint-action}
\operatorname{Ad}_g(X)=gXg^\dagger.
\end{equation}
The matrix \(R(g)\) defined in Eq.~\eqref{w4:eq:adjoint-matrix} represents this action
in the fixed orthonormal Hermitian basis \(T_0,\ldots,T_7\).
Write \(|X\rangle\in W_8\) for the coordinate ket of \(X\) in this basis.
Composing coordinate reconstruction \(|X\rangle\mapsto X\) with matrix vectorization gives the fixed isometry:
\begin{equation}\label{w4:eq:adjoint-coordinate-embedding}
{A_8^{3,\bar3}}:W_8\hookrightarrow V\otimes V^*,\qquad
{A_8^{3,\bar3}}|X\rangle=|X\rangle\!\rangle
=\sum_{b,c=0}^2X_{bc}|bc\rangle.
\end{equation}
Its image is the traceless matrix subspace of the nine-dimensional
tensor product. The conjugate ket carrier \(V^*\) is defined in
Sec.~\ref{w4:sec:notation}; bars on numerical representation labels
denote conjugate representations. We use the coordinate copy in the fixed bases \(V\to V^*\) that sends each computational basis
ket to the corresponding basis ket and preserves its coefficients.
It is a linear coordinate map, whereas the ket-to-bra Riesz
correspondence is antilinear.
In labeled tensor products, upright subscripts name the factors,
while a numerical subscript on \(I_d\) specifies a dimension.
Products whose output labels are disjoint are placed in the stated
factor order.

The catalyst \(|c_{4/5}(U)\rangle_{\mathsf P}\), defined in
Eq.~\eqref{w4:eq:catalyst-specific}, carries \((1,1)\oplus(8,8)\), with the
left and right labels specified in the \hyperref[w4:tc:representation-labels]{Notation paragraph} at the start of Sec.~\ref{w4:chap:covariance}.
The two \(8\) labels refer to its first and second adjoint factors.
A fixed transpose puts the second factor's right action into standard
adjoint form. We apply \(\mathsf T: W_8\to W_8\) to that factor, with
\(\mathsf T(X)=X^\top\) and
\(\mathsf T_{rs}=\operatorname{tr}(T_r^\dagger T_s^\top)\).
The map preserves the Hilbert--Schmidt inner product, and
\((h^*Xh^\top)^\top=hX^\top h^\dagger\) gives the required right action.
The data contribute the left action \(g^*\), with label \(\bar3\).
Coupling them to the first catalyst factor gives
\begin{equation}\label{w4:eq:initial-representation}
\begin{aligned}
(\bar3,1)\otimes[(1,1)\oplus(8,8)]
&=(\bar3\otimes1,1)\oplus(\bar3\otimes8,8)\\
&=(\bar3,1)\oplus(\bar3\oplus6\oplus\overline{15},8)\\
&=(\bar3,1)\oplus(\bar3,8)\oplus(6,8)
  \oplus(\overline{15},8).
\end{aligned}
\end{equation}
Thus only the first adjoint factor participates in the left
decomposition; the second remains the right \(8\) in each block.

Two fixed isometric embeddings realize the left decomposition,
with the conjugate fundamental space \(V^*\) carrying \(g^*\):
\({A_{\bar3}^{\bar3,8}}:V^*\hookrightarrow V^*\otimes W_8\) and
\({A_6^{\bar3,8}}:\operatorname{Sym}^2V\hookrightarrow V^*\otimes W_8\).
Label the initial data factor by \(\mathsf D\) and the two matrix
factors by \(\mathsf B,\mathsf C\), in that order.
Vectorizing the \(W_8\) factor gives \(27\) ambient coordinates for
the \(24\)-dimensional target \(V^*\otimes W_8\).
Qutrit coordinate indices take values \(0,1,2\), whereas the labels
\(r,s\) in \(T_r,T_s\) range from \(0\) to \(7\).
Coordinate indices are local to each tensor formula. Sums are written
explicitly; an index without a summation sign is a free entry label.
Let \(\delta_{ij}\) be the Kronecker delta and \(\epsilon_{ijk}\)
the completely antisymmetric Levi--Civita tensor, with
\(\epsilon_{012}=1\). A symmetric matrix \(Y\) represents the ket
\(|Y\rangle=\sum_{a,c=0}^2Y_{ac}|ac\rangle\in\operatorname{Sym}^2V\).
Using the natural inclusion \({A_6^{3,3}}\) defined above,
the two embeddings have the vectorized operator forms
\begin{equation}\label{w4:eq:initial-embeddings}
\begin{aligned}
(I_{\mathsf D}\otimes{A_8^{3,\bar3}}){A_{\bar3}^{\bar3,8}}
&=\sqrt{\frac38}\sum_{j=0}^2
\left(|I\rangle\!\rangle_{\mathsf{DB}}|j\rangle_{\mathsf C}
-\frac13|j\rangle_{\mathsf D}|I\rangle\!\rangle_{\mathsf{BC}}\right)
\langle j|,\\
(I_{\mathsf D}\otimes{A_8^{3,\bar3}}){A_6^{\bar3,8}}
&=\left(\frac1{\sqrt2}\sum_{j,a,b,c=0}^2
\epsilon_{jbc}|j\rangle_{\mathsf D}|a\rangle_{\mathsf B}
|b\rangle_{\mathsf C}\langle ac|\right){A_6^{3,3}}.
\end{aligned}
\end{equation}
Here \(|I\rangle\!\rangle_{\mathsf{DB}}
=\sum_{j=0}^2|j\rangle_{\mathsf D}|j\rangle_{\mathsf B}\).
The two terms in \({A_{\bar3}^{\bar3,8}}\) remove the trace on
\(\mathsf B,\mathsf C\), and symmetry of the input to \({A_6^{\bar3,8}}\)
makes its image traceless. Contracting these matrix factors with
\(\langle\!\langle T_r|\) gives the intrinsic \(W_8\) coordinates.
Equivalently,
\[
(I_{\mathsf D}\otimes{A_8^{3,\bar3}}){A_6^{\bar3,8}}
=(\mathrm{SWAP}\otimes I_3)\Bigl(I_3\otimes\bigl(A_{\bar3}^{3,3}\bigr)^*\Bigr)
{A_6^{3,3}}.
\]
Since \({A_{\bar3}^{3,3}}\) is real, conjugating Eq.~\eqref{w4:eq:iota} gives
\((g^*\otimes g^*)\bigl(A_{\bar3}^{3,3}\bigr)^*=\bigl(A_{\bar3}^{3,3}\bigr)^* g\); its same coefficient matrix
is used here as \(V\to V^*\otimes V^*\).
The embedding \({A_6^{\bar3,8}}\) acts on
\(\operatorname{Sym}^2V\subset V\otimes V\), of dimension six;
the tensor operator above has two input factors and three
output factors. Extending \({A_6^{\bar3,8}}\) by zero on the antisymmetric input
is achieved by replacing \({A_6^{3,3}}\) with
\(\Pi_{\operatorname{Sym}^2}\). The Gram operator of this extension
is \(\Pi_{\operatorname{Sym}^2}\), whereas
\(\bigl(A_6^{\bar3,8}\bigr)^\dagger {A_6^{\bar3,8}}=I_6\) on its full symmetric domain.
The tensor coordinates give the norm
\(\|Y\|^2=\sum_{a,c=0}^2|Y_{ac}|^2\). Its orthonormal ket basis consists of
\(|aa\rangle\), \(a=0,1,2\), and
\((|ac\rangle+|ca\rangle)/\sqrt2\), \(0\le a<c\le2\).
Paired indices such as \(ac\) label this natural tensor inclusion;
the normalized six-dimensional basis uses a single row index.
Contractions of \(\delta\) and \(\epsilon\) give
\[
\bigl(A_{\bar3}^{\bar3,8}\bigr)^\dagger {A_{\bar3}^{\bar3,8}}=I_3,\qquad
\bigl(A_6^{\bar3,8}\bigr)^\dagger {A_6^{\bar3,8}}=I_6,\qquad \bigl(A_{\bar3}^{\bar3,8}\bigr)^\dagger {A_6^{\bar3,8}}=0.
\]
The invariant Bell pairing and Levi--Civita tensor give
\[
(g^*\otimes\operatorname{Ad}_g){A_{\bar3}^{\bar3,8}}={A_{\bar3}^{\bar3,8}}g^*,
\qquad
(g^*\otimes\operatorname{Ad}_g){A_6^{\bar3,8}}={A_6^{\bar3,8}}\operatorname{Sym}^2g,
\]
where \(\operatorname{Ad}_g\) is defined in Eq.~\eqref{w4:eq:adjoint-action}.
We construct the remaining inclusion
\({A_{\overline{15}}^{\bar3,8}}:W_{\overline{15}}\hookrightarrow
V^*\otimes W_8\) by projecting the ordered vectors
\(|i\rangle\otimes|T_r\rangle\) with
\(I_{24}-{A_{\bar3}^{\bar3,8}}\bigl(A_{\bar3}^{\bar3,8}\bigr)^\dagger-{A_6^{\bar3,8}}\bigl(A_6^{\bar3,8}\bigr)^\dagger\) and applying
Gram--Schmidt to the nonzero independent projections. Its fifteen
orthonormal columns define the remaining representation by
\[
\overline{15}(g)=\bigl(A_{\overline{15}}^{\bar3,8}\bigr)^\dagger
(g^*\otimes\operatorname{Ad}_g){A_{\overline{15}}^{\bar3,8}}.
\]
Thus \(\overline{15}\) denotes the remaining irreducible component
of \(\bar3\otimes8\); the label \(15\) used below denotes its
complex-conjugate representation.

For \(\lambda=\bar3,6,\overline{15}\), let
\(d_\lambda=3,6,15\), with \(W_{\bar3}=V^*\) and
\(W_6=\operatorname{Sym}^2V\). Projecting the left factor of the
adjoint Bell component through \(\bigl(A_\lambda^{\bar3,8}\bigr)^\dagger\) gives the initial
seed \(S_{\lambda,8}:V\to W_\lambda\otimes W_8\), defined by
\begin{equation}\label{w4:eq:initial-block-seeds}
S_{\lambda,8}=\sqrt{\frac3{d_\lambda}}
\Bigl(\bigl(A_\lambda^{\bar3,8}\bigr)^\dagger\otimes I_8\Bigr)
\left[I_V\otimes\sum_{r=0}^7
|T_r\rangle\otimes|T_r^\top\rangle\right].
\end{equation}
The transpose in the last factor comes from \(\mathsf T\) on the second
catalyst factor. If \(\Pi_\lambda={A_\lambda^{\bar3,8}} \bigl(A_\lambda^{\bar3,8}\bigr)^\dagger\),
Schur's lemma gives
\(\operatorname{Tr}_8\Pi_\lambda=(d_\lambda/3)I_3\), which proves
\(S_{\lambda,8}^\dagger S_{\lambda,8}=I_3\). Each seed therefore
maps a normalized data vector to a normalized branch state. The reduced state of
the normalized adjoint Bell component is \(I_8/8\); hence its
projection to a left component of dimension \(d_\lambda\) has
probability \((4/5)d_\lambda/24\). In the block order
\((\bar3,1),(\bar3,8),(6,8),(\overline{15},8)\), the four amplitudes are
\begin{equation}\label{w4:eq:initial-weights}
\frac1{\sqrt5},\qquad \frac1{\sqrt{10}},\qquad
\frac1{\sqrt5},\qquad \frac1{\sqrt2}.
\end{equation}
Seed reconstruction now gives the input {map }in representation coordinates
after the fixed transpose on the second catalyst factor and the
left decomposition by \(\bigl(A_\lambda^{\bar3,8}\bigr)^\dagger\):
\begin{equation}\label{w4:eq:initial-state}
\frac1{\sqrt5}{I_{\mathsf S}} \oplus\frac1{\sqrt{10}}(U^*\otimes I_8)S_{\bar3,8}U^\top \oplus\frac1{\sqrt5} (\operatorname{Sym}^2U\otimes I_8)S_{6,8}U^\top \oplus\frac1{\sqrt2} (\overline{15}(U)\otimes I_8)S_{\overline{15},8}U^\top.
\end{equation}
All direct sums below use fixed representation coordinates embedded
in the physical workspace~\eqref{w4:eq:workspace}.{ In the trajectory formulas, \(\oplus\) stacks branch maps with the common input space \(V\), as in Eq.~\eqref{w4:eq:phi-family}.}

\subsection{The two-query encoder}

\begin{wlemma}[Two-query encoder]\label{w4:lem:two-query}
There are fixed unitary gates \(G_0,G_1,G_2\) on \(\mathcal H\) such
that the two-query circuit
\begin{equation}\label{w4:eq:two-query-circuit}
\mathcal E(U)=G_2(U_{\mathsf S}\otimes I_{\mathsf A})
G_1(U_{\mathsf S}\otimes I_{\mathsf A})G_0
\end{equation}
satisfies, for every \(U\in SU(3)\),
\begin{equation}\label{w4:eq:two-query-endpoint}
\mathcal E(U)({I_{\mathsf S}}\otimes
|c_{4/5}(U)\rangle_{\mathsf A})=A_{\mathcal M_4}^{\mathcal H}\Phi_{1/5}(U).
\end{equation}
\end{wlemma}

\begin{wproof}[Proof of Lemma~\ref{w4:lem:two-query}]
The first query selects three occupied representation blocks, and the
second maps them to the intermediate family. We fix the encoding and
decoding phases through their seeds and extend the resulting isometries
to unitaries on the common workspace.

\subsubsection{First encoding and query}

For an input block \((\lambda,\sigma)\), a right embedding
\(\sigma\hookrightarrow\bar3\otimes\mu\) exposes the query
qutrit with action \(h^*\); the memory retains \((\lambda,\mu)\). This fixed intertwining encoding preserves covariance by Lemma~\ref{w4:lem:tc-encoding}.
After a forward query, the {output action is \(g\otimes\lambda(g)\otimes\mu(h)\)} by Lemma~\ref{w4:lem:tc-query}. {By Eq.~\eqref{w4:eq:query-selection-rule},
the nonzero output branches therefore satisfy}
\[
\lambda_{\mathrm{out}}\in
\operatorname{Irr}(3\otimes\lambda)
\cap\operatorname{Irr}(\bar3\otimes\mu).
\]
The right memories \(3,\bar6\) at this query and \(6,\bar3\)
at the next use the computational or normalized symmetric tensor
bases, which satisfy the compatibility condition
in Eq.~\eqref{w4:tc:compatible-basis}.
The Hermitian basis of \(W_8\) is used for adjoint coordinates,
with its tensor action expressed through \({A_8^{3,\bar3}}\).

For \(\lambda=6,\overline{15}\), Bell pairing gives the embedding
\({A_8^{3,\lambda}}: W_8\hookrightarrow V\otimes W_\lambda\),
\begin{equation}\label{w4:eq:left-eight-embedding}
{A_8^{3,\lambda}}=\sqrt{\frac8{d_\lambda}}
\Bigl(I_{\mathsf S}\otimes \bigl(A_\lambda^{\bar3,8}\bigr)^\dagger\Bigr)
(|I\rangle\!\rangle_{\mathsf{SD}}\otimes I_{W_8}).
\end{equation}
This map intertwines \(8\) with \(3\otimes\lambda\).
The Bell ket creates \(\mathcal H_{\mathsf S}\otimes(V^*)_{\mathsf D}\),
and \(\bigl(A_\lambda^{\bar3,8}\bigr)^\dagger\) contracts
\((V^*)_{\mathsf D}\otimes W_8\), leaving the output
\(\mathcal H_{\mathsf S}\otimes W_\lambda\). The identity
\(\operatorname{Tr}_{V^*}\Pi_\lambda=(d_\lambda/8)I_8\) gives
\(\bigl(A_8^{3,\lambda}\bigr)^\dagger {A_8^{3,\lambda}}=I_8\). For the displayed \({A_6^{\bar3,8}}\),
\[
[{A_8^{3,6}}|X\rangle]_{q,ac}=-\frac1{\sqrt6}\sum_{t=0}^2
(\epsilon_{qat}X_{ct}+\epsilon_{qct}X_{at}).
\]
For the singlet and adjoint right factors, use
the normalized Bell inclusion \(\mathbb C\hookrightarrow V^*\otimes V\) and
\(\bigl(A_8^{3,\bar3}\bigr)^*\), respectively, where \({A_8^{3,\bar3}}\) is defined in Eq.~\eqref{w4:eq:adjoint-coordinate-embedding}. The former is defined for \(z\in\mathbb C\) by
\begin{equation}\label{w4:eq:first-right-embeddings}
z\longmapsto\frac{z}{\sqrt3}|I\rangle\!\rangle_{\mathsf{SR}}.
\end{equation}
For a right \(\bar6\) memory, use
\(\bigl[-(A_8^{3,6})^*\bigr]: W_8\hookrightarrow
V^*\otimes\operatorname{Sym}^2(V^*)\), with
\begin{equation}\label{w4:eq:first-six-encoding}
\Bigl[\bigl[-(A_8^{3,6})^*\bigr]|X\rangle\Bigr]_{i,jk}=\frac1{\sqrt6}\sum_{t=0}^2
\bigl(\epsilon_{ijt}X_{tk}+\epsilon_{ikt}X_{tj}\bigr).
\end{equation}
The conjugate embedding \(\bigl(A_8^{3,\bar3}\bigr)^*\) places the \(V^*\) query factor before the \(V\) memory factor.
The normalized Bell inclusion and \(\bigl(A_8^{3,\bar3}\bigr)^*\) are orthogonal isometries onto the singlet and adjoint
parts of \(\bar3\otimes3\). The squared norm of the numerator in
Eq.~\eqref{w4:eq:first-six-encoding} is
\(6\|X\|^2-2|\operatorname{tr}X|^2\), so this map is isometric on
\( W_8\). Their intertwining relations are
\[
(h^*\otimes h)\bigl(A_8^{3,\bar3}\bigr)^*=\bigl(A_8^{3,\bar3}\bigr)^*\operatorname{Ad}_h,
\qquad
(h^*\otimes\operatorname{Sym}^2h^*)\bigl[-(A_8^{3,6})^*\bigr]
=\bigl[-(A_8^{3,6})^*\bigr]\operatorname{Ad}_h.
\]
We encode with the Bell inclusion and \(\bigl(A_8^{3,\bar3}\bigr)^*\) for the initial \((\bar3,1),(\bar3,8)\)
blocks, \(\bigl[-(A_8^{3,6})^*\bigr]\) for \((6,8)\), and \(\bigl(A_8^{3,\bar3}\bigr)^*\) for
\((\overline{15},8)\). The first-query memory is
\begin{equation}\label{w4:eq:first-memory}
\mathcal M_1=(V^*\otimes V)\oplus
(\operatorname{Sym}^2V\otimes\operatorname{Sym}^2(V^*))
\oplus(W_{\overline{15}}\otimes V),
\qquad \dim\mathcal M_1=9+36+45=90.
\end{equation}
We identify \(\mathcal M_1\) with \(\mathcal H_{\mathsf A}\) through fixed
orthonormal bases and place the query qutrit first by permuting factors. Together with \(\mathsf T\) and the initial left
decomposition, these maps specify \(G_0\) isometrically on the full
\(195\)-dimensional input subspace \(\mathcal H_{\mathsf S}\otimes\mathcal H_{\mathsf P}\).

For the first memory branch, label the three factors as query
\(\mathsf S\), left memory \(\mathsf L\), and right memory
\(\mathsf R\), in that order. Define the seed maps from \(V\) to
\(\mathcal H_{\mathsf S}\otimes(V^*)_{\mathsf L}\otimes V_{\mathsf R}\)
by
\[
u_1=
{\frac1{\sqrt3}\sum_{a,j=0}^2|a\rangle_{\mathsf S}|a\rangle_{\mathsf L}|j\rangle_{\mathsf R}\langle j|},
\qquad
v_1=
{\frac1{\sqrt3}\sum_{a,j=0}^2|a\rangle_{\mathsf S}|j\rangle_{\mathsf L}|a\rangle_{\mathsf R}\langle j|},
\qquad w_1=\frac{3u_1-v_1}{\sqrt8}.
\]
{The bras use the fixed input basis of \(V\). }{The maps} satisfy \(u_1^\dagger u_1=v_1^\dagger v_1=I_3\),
\(u_1^\dagger v_1=I_3/3\), and \(v_1^\dagger w_1=0\),
\(w_1^\dagger w_1=I_3\). Define the fixed seeds for the other two
branches by
\[
\chi_6=(\mathrm{SWAP}_{6,3}\otimes I_6)
\Bigl(I_6\otimes \bigl[-(A_8^{3,6})^*\bigr]\Bigr)S_{6,8},
\qquad
\chi_{\overline{15}}=(\mathrm{SWAP}_{15,3}\otimes I_3)
\Bigl(I_{15}\otimes \bigl(A_8^{3,\bar3}\bigr)^*\Bigr)S_{\overline{15},8},
\]
where \(S_{6,8}\) and \(S_{\overline{15},8}\) are defined in Eq.~\eqref{w4:eq:initial-block-seeds}.
Each map takes \(V\) into its corresponding sector of
\(\mathcal H_{\mathsf S}\otimes\mathcal M_1\), where \(\mathcal M_1\) is defined in Eq.~\eqref{w4:eq:first-memory}. The encoded {map} before
the first query is
\begin{equation}\label{w4:eq:first-query-before}
(I_{\mathsf S}\otimes U_{\mathsf L}^*\otimes I_{\mathsf R}) \left(\frac{v_1}{\sqrt5}+\frac{w_1}{\sqrt{10}}\right)U^\top \oplus\frac1{\sqrt5} (I_{\mathsf S}\otimes(\operatorname{Sym}^2U)_{\mathsf L}\otimes I_{\mathsf R}) \chi_6U^\top \oplus\frac1{\sqrt2} (I_{\mathsf S}\otimes\overline{15}(U)_{\mathsf L}\otimes I_{\mathsf R}) \chi_{\overline{15}}U^\top.
\end{equation}
The first query replaces the first identity in each summand by
\(U\); in the first summand this gives
\((U_{\mathsf S}\otimes U_{\mathsf L}^*\otimes I_{\mathsf R})
(v_1/\sqrt5+w_1/\sqrt{10})U^\top\).

\noindent\begin{minipage}{\textwidth}
The selection rule gives the following output blocks:
\begin{center}
\setlength{\tabcolsep}{3pt}
\renewcommand{\arraystretch}{1.25}
\begin{tabularx}{\linewidth}{@{}l X X l@{}}
\toprule
Initial blocks & Right embedding & Left candidates & Output blocks\\
\midrule
\((\bar3,1),(\bar3,8)\) &
\(1\oplus8=\bar3\otimes3\) &
\(3\otimes\bar3=1\oplus8\) & \((1,3),(8,3)\)\\
\((6,8)\) & \(8\hookrightarrow\bar3\otimes\bar6=8\oplus\overline{10}\) &
\(3\otimes6=8\oplus10\) & \((8,\bar6)\)\\
\((\overline{15},8)\) & \(8\hookrightarrow\bar3\otimes3=1\oplus8\) &
\(3\otimes\overline{15}=8\oplus\overline{10}\oplus27\) & \((8,3)\)\\
\bottomrule
\end{tabularx}
\end{center}
\end{minipage}
\par
The right memories \(3,\bar6\) permit, respectively,
\(\bar3\otimes3=1\oplus8\) and
\(\bar3\otimes\bar6=8\oplus\overline{10}\). In particular,
\(10=\operatorname{Sym}^3V\) and its complex conjugate \(\overline{10}\) are inequivalent
representations. The label \(27\) denotes the
remaining irreducible component, of dimension \(3\cdot15-8-10=27\). Every unwanted branch has zero seed and therefore
vanishes for all \(U\), by Proposition~\ref{w4:lem:tc-seed}.

\subsubsection{A worked branch: \((6,8)\) to \((8,\bar6)\)}
\label{w4:sec:first-query-example}
\label{w4:ex:first-6-8}

The \((6,8)\) branch leaves only \((8,\bar6)\) after the first query.
Its left \(6\) comes from coupling the data with the catalyst's first
adjoint factor; its right \(8\) is the second adjoint factor.
The initial {map}, its encoding, and the queried {map} are
\begin{equation}\label{w4:eq:worked-six-branch}
\frac1{\sqrt5}(\operatorname{Sym}^2U\otimes I_8) S_{6,8}U^\top, \frac1{\sqrt5}(I_{\mathsf S}\otimes(\operatorname{Sym}^2U)_{\mathsf L}\otimes I_{\mathsf R}) \chi_6U^\top, \frac1{\sqrt5}(U_{\mathsf S}\otimes(\operatorname{Sym}^2U)_{\mathsf L}\otimes I_{\mathsf R}) \chi_6U^\top.
\end{equation}
The last {map} has left candidates \(8\oplus10\) and right memory
\(\bar6\). The seed constraint permits \(8\oplus\overline{10}\),
so only the left \(8\) survives.

We decode the surviving left adjoint component with the embedding
\({A_8^{3,6}}\) defined in Eq.~\eqref{w4:eq:left-eight-embedding}
and fix its phase using the standard normalized seed
\(S_{8,\bar6}:V\to W_8\otimes\operatorname{Sym}^2(V^*)\),
whose entries below are those of
\(({A_8^{3,\bar3}}\otimes I_{\operatorname{Sym}^2(V^*)})S_{8,\bar6}\),
with the symmetric factor written in its natural tensor coordinates:
\begin{equation}\label{w4:eq:seed-eight-six}
(S_{8,\bar6})_{bc,ad;j}
=\frac14(\delta_{ac}\epsilon_{bdj}-\epsilon_{abj}\delta_{cd}).
\end{equation}
Here \(bc\) are traceless matrix indices, \(ad\) are symmetric
right-memory indices, and \(j\) is the data column. At \(U=I\),
contraction with the uncorrected \(\bigl(A_8^{3,6}\bigr)^\dagger\) gives
\[
\frac14(\epsilon_{jab}\delta_{cd}+\epsilon_{jdb}\delta_{ca})
=-(S_{8,\bar6})_{bc,ad;j}.
\]
We therefore choose the fixed decoding
\(-\bigl(A_8^{3,6}\bigr)^\dagger\). For arbitrary \(U\), it produces
\begin{equation}\label{w4:eq:worked-six-output}
\frac1{\sqrt5}(R(U)\otimes I_6)S_{8,\bar6}U^\top
{:V\longrightarrow} W_8\otimes\operatorname{Sym}^2(V^*).
\end{equation}
This space carries type \((8,\bar6)\).
Its amplitude is preserved because the encoding, query, and decoding
are isometric on this single surviving branch. The explicit
second-query calculation below also verifies
\(S_{8,\bar6}^\dagger S_{8,\bar6}=I_3\).

\subsubsection{First recoupling and the fixed multiplicity rotation}

Recoupling the first memory branch preserves its chosen amplitudes.
The source basis \((v_1,w_1)\) and the post-query singlet--adjoint basis
\((u_1,(3v_1-u_1)/\sqrt8)\) satisfy
\[
\left(u_1,\frac{3v_1-u_1}{\sqrt8}\right)=(v_1,w_1)O_1,
\]
where
\begin{equation}\label{w4:eq:first-recoupling}
O_1=\frac13\begin{pmatrix}1&\sqrt8\\\sqrt8&-1\end{pmatrix},
\qquad
O_1\binom{1/\sqrt5}{1/\sqrt{10}}
=\binom{1/\sqrt5}{1/\sqrt{10}}.
\end{equation}
Because \(O_1\) is real, orthogonal, and symmetric, it also maps the
source coefficient column to the target coefficient column, leaving
the chosen amplitudes unchanged. At the tensor level,
\begin{equation}\label{w4:eq:first-seed-match}
\sqrt2v_1+\frac{3u_1-v_1}{\sqrt8}
=\frac{3(u_1+v_1)}{\sqrt8}
=\sqrt2u_1+\frac{3v_1-u_1}{\sqrt8}.
\end{equation}
The \((1,3)\) decoding contracts \(\mathsf S,\mathsf L\) with the
normalized Bell bra. The \((8,3)\) decoding contracts these factors
with \(\langle\!\langle T_r|\), giving the standard seed
\(S_{8,3}:V\to W_8\otimes V\). Its matrix-coordinate entries below
are those of \(({A_8^{3,\bar3}}\otimes I_V)S_{8,3}\):
\begin{equation}\label{w4:eq:seed-eight-three}
(S_{8,3})_{bc,k;j}=\sqrt{\frac38}
\left(\delta_{bk}\delta_{cj}-\frac13\delta_{bc}\delta_{kj}\right).
\end{equation}
The last first-query branch uses \(\bigl(A_8^{3,\overline{15}}\bigr)^\dagger\) from
Eq.~\eqref{w4:eq:left-eight-embedding}. Its decoded seed is also the
positive \(S_{8,3}\): contracting the projected initial Bell seed
uses
\(\sum_{r=0}^7(T_r)_{bc}(T_r)_{jk}
=\delta_{bk}\delta_{cj}-\delta_{bc}\delta_{jk}/3\).
These tensor contractions fix every decoding phase. Together they define \(\mathcal D_1\), an isometric decoding
from the 99-dimensional surviving first-query subspace to
\[
V\oplus(W_8\otimes V)\oplus
(W_8\otimes\operatorname{Sym}^2(V^*))\oplus(W_8\otimes V).
\]
The four summands have types
\((1,3),(8,3),(8,\bar6),(8,3)\), respectively. On this source, the ordered branch contractions are
\begin{equation}\label{w4:eq:first-decoding}
\mathcal D_1
=\left[
\begin{pmatrix}
\langle\!\langle I|/\sqrt3\\
\bigl(A_8^{3,\bar3}\bigr)^\dagger
\end{pmatrix}\otimes I_3\right]
\oplus\left[-\bigl(A_8^{3,6}\bigr)^\dagger\otimes I_6\right]
\oplus\left[\bigl(A_8^{3,\overline{15}}\bigr)^\dagger\otimes I_3\right].
\end{equation}
The two stacked contractions in the first block share the same
\(27\)-dimensional source. Its output is
\begin{equation}\label{w4:eq:first-decoded-state}
\frac1{\sqrt5}U^\top \oplus\frac1{\sqrt{10}}(R(U)\otimes I_3)S_{8,3}U^\top \oplus\frac1{\sqrt5}(R(U)\otimes I_6)S_{8,\bar6}U^\top \oplus\frac1{\sqrt2}(R(U)\otimes I_3)S_{8,3}U^\top.
\end{equation}
A fixed multiplicity rotation merges the two \((8,3)\) copies.
Their seeds and output actions agree, so they carry the same normalized
family for every \(U\). We rotate their multiplicity factor by
\begin{equation}\label{w4:eq:multiplicity-rotation}
O_{\mathrm{merge}}=\frac1{\sqrt6}
\begin{pmatrix}1&\sqrt5\\\sqrt5&-1\end{pmatrix},
\qquad
O_{\mathrm{merge}}\binom{1/\sqrt{10}}{1/\sqrt2}
=\binom{\sqrt{3/5}}0.
\end{equation}
The occupied family now lies in
\begin{equation}\label{w4:eq:first-occupied-space}
\mathcal M_2=V\oplus(W_8\otimes V)
\oplus(W_8\otimes\operatorname{Sym}^2(V^*)),
\qquad \dim\mathcal M_2=3+24+48=75,
\end{equation}
with block types \((1,3),(8,3),(8,\bar6)\) and {map}
\begin{equation}\label{w4:eq:first-query-amplitudes}
\frac1{\sqrt5}\left[
U^\top\oplus
\sqrt3(R(U)\otimes I_3)S_{8,3}U^\top\oplus
(R(U)\otimes I_6)S_{8,\bar6}U^\top\right].
\end{equation}
The rotation acts unitarily on both full copies and places this
{map's image} in the first.

\subsubsection{Second encoding and query}

For the second query we use the memory
\begin{equation}\label{w4:eq:second-memory}
\mathcal M_3=\operatorname{Sym}^2V\oplus( W_8\otimes V^*),
\qquad \dim\mathcal M_3=6+24=30,
\end{equation}
embedded into the same \(90\)-dimensional physical memory.
The \((1,3)\) branch reaches its target through the encoding
\({\Phi_{3,6}}(I):
V\hookrightarrow V^*\otimes\operatorname{Sym}^2V\),
which realizes \(3\hookrightarrow\bar3\otimes6\).
It is isometric by the normalization in Sec.~\ref{w4:sec:notation} and
satisfies \((h^*\otimes\operatorname{Sym}^2h){\Phi_{3,6}}(I)
={\Phi_{3,6}}(I)h\). Its {map} before and after the second query is
\begin{equation}\label{w4:eq:second-singlet-state}
\frac1{\sqrt5}{\Phi_{3,6}}(I)U^\top
\ \longmapsto\
\frac1{\sqrt5}(U\otimes I_6){\Phi_{3,6}}(I)U^\top
=\frac1{\sqrt5}{\Phi_{3,6}}(U).
\end{equation}

The two left-adjoint blocks enter orthogonal antisymmetric and symmetric
sectors. We order their tensor factors as query \(\mathsf S\), matrix
coordinates \(\mathsf B,\mathsf C\), and right memory \(\mathsf R\),
and vectorize the matrix ket as \(|X\rangle\!\rangle_{\mathsf{BC}}\).
On \(\mathsf S,\mathsf R\), we encode right \(3\) with
\(-\bigl(A_{\bar3}^{3,3}\bigr)^*\) and right \(\bar6\) with the same real symmetric inclusion
\(A_{\bar6}^{\bar3,\bar3}\), now of type
\(\operatorname{Sym}^2(V^*)\hookrightarrow V^*\otimes V^*\),
on these two factors.
The antisymmetric and symmetric images are orthogonal and together
fill the \(72\)-dimensional space associated with the common
\(24\)-dimensional memory. These embeddings and
\({\Phi_{3,6}}(I)\) define an isometry
\(\mathcal E_2:\mathcal M_2\to \mathcal H_{\mathsf S}\otimes\mathcal M_3\).
{In these input and output decompositions, its block matrix is}
\begin{equation}\label{w4:eq:second-encoding}
\mathcal E_2={\begin{pmatrix}
\Phi_{3,6}(I)&0&0\\
0&-(\mathrm{SWAP}_{8,3}\otimes I_3)\bigl(I_8\otimes(A_{\bar3}^{3,3})^*\bigr)&
(\mathrm{SWAP}_{8,3}\otimes I_3)\bigl(I_8\otimes A_{\bar6}^{\bar3,\bar3}\bigr)
\end{pmatrix}}.
\end{equation}
The swap places the query factor first; the domain has dimension
\(75\), and the ambient codomain has dimension \(90\).

Let
\(\Pi_{\mathrm{tr}}=I_{\mathsf B}\otimes I_{\mathsf C}
-|I\rangle\!\rangle_{\mathsf{BC}}\langle\!\langle I|_{\mathsf{BC}}/3\)
be the projector onto \({A_8^{3,\bar3}}(W_8)\).
Define \(u_2,v_2:V\to\mathcal H_{\mathsf S}\otimes W_8\otimes(V^*)_{\mathsf R}\)
through their vectorized forms:
\begin{equation}\label{w4:eq:second-pairings}
\begin{aligned}
(I_{\mathsf S}\otimes{A_8^{3,\bar3}}\otimes I_{\mathsf R})u_2
&=\sqrt{\frac38}(I_{\mathsf S}\otimes\Pi_{\mathrm{tr}}\otimes I_{\mathsf R})
\Bigl(|I\rangle\!\rangle_{\mathsf{SC}}\otimes\bigl(A_{\bar3}^{3,3}\bigr)_{\mathsf{BR}}\Bigr),\\
(I_{\mathsf S}\otimes{A_8^{3,\bar3}}\otimes I_{\mathsf R})v_2
&=\sqrt{\frac38}(I_{\mathsf S}\otimes\Pi_{\mathrm{tr}}\otimes I_{\mathsf R})
\Bigl(\bigl(A_{\bar3}^{3,3}\bigr)_{\mathsf{SB}}\otimes|I\rangle\!\rangle_{\mathsf{CR}}\Bigr).
\end{aligned}
\end{equation}
Both right sides are placed in the order
\(\mathsf S,\mathsf B,\mathsf C,\mathsf R\).
Write \((u_2)_{abcd;j}\) and \((v_2)_{abcd;j}\) for the entries of these
vectorized maps; they are
\begin{equation}\label{w4:eq:second-pairing-coordinates}
\begin{aligned}
(u_2)_{abcd;j}&=\frac{\sqrt3}{4}
\left(\delta_{ac}\epsilon_{bdj}-\frac13\delta_{bc}\epsilon_{adj}\right),\\
(v_2)_{abcd;j}&=\frac{\sqrt3}{4}
\left(\epsilon_{abj}\delta_{cd}-\frac13\delta_{bc}\epsilon_{adj}\right).
\end{aligned}
\end{equation}
Both pairings are isometries with a fixed cross overlap; Bell contraction gives
\begin{equation}\label{w4:eq:second-pairing-gram}
u_2^\dagger u_2=v_2^\dagger v_2=I_3,\qquad u_2^\dagger v_2=-I_3/2.
\end{equation}
For example, the unprojected cross contraction is
\(\frac12\sum_{a,b=0}^2\epsilon_{baj}\epsilon_{abk}=-\delta_{jk}\);
removing the common Bell direction changes the diagonal and cross
entries to \(8/3\) and \(-4/3\), respectively, before the common
normalization \(3/8\).

The right-\(3,\bar6\) source seeds are the orthogonal isometries
\begin{equation}\label{w4:eq:second-seeds}
u_2+v_2,\qquad \frac{u_2-v_2}{\sqrt3}.
\end{equation}
Encoding \(S_{8,3}\) with \(-\bigl(A_{\bar3}^{3,3}\bigr)^*\) gives the first, and encoding
\(S_{8,\bar6}\) symmetrically gives the second. These identities
follow directly from Eq.~\eqref{w4:eq:second-pairing-coordinates} and
the three-dimensional identity
\begin{equation}\label{w4:eq:epsilon-recoupling}
\delta_{ac}\epsilon_{bdj}-\delta_{bc}\epsilon_{adj}
=\epsilon_{abd}\delta_{cj}-\epsilon_{abj}\delta_{cd}.
\end{equation}
They also verify the normalization of both standard seeds.

Define the fixed source combination
\begin{equation}\label{w4:eq:second-source-combination}
\Xi=\sqrt{\frac35}(u_2+v_2)+\frac1{\sqrt5}\frac{u_2-v_2}{\sqrt3}.
\end{equation}
The second query acts on the left-adjoint part as
\begin{equation}\label{w4:eq:second-query-states}
(I_{\mathsf S}\otimes R(U)\otimes I_{\mathsf R})\Xi U^\top
\longmapsto
(U_{\mathsf S}\otimes R(U)\otimes I_{\mathsf R})\Xi U^\top.
\end{equation}
Here
\(\bigl(A_8^{3,\bar3}\bigr)^\dagger(U_{\mathsf B}\otimes U_{\mathsf C}^*)
{A_8^{3,\bar3}}=R(U)\), so the {image of the displayed map} lies in
the \(3\times8\times3=72\)-dimensional carrier.
Thus the post-query left candidates are
\(3\otimes8=3\oplus\bar6\oplus15\). Right memory \(\bar3\)
permits only \(\bar3\otimes\bar3=3\oplus\bar6\), so the
\(15\) branch has zero seed and vanishes for all \(U\).

\subsubsection{Second recoupling and the final decoding}

The second recoupling also preserves the chosen amplitudes.
Its normalized target basis is \((u_2,(u_2+2v_2)/\sqrt3)\).
The Gram relations in Eq.~\eqref{w4:eq:second-pairing-gram} connect it
to the source basis in Eq.~\eqref{w4:eq:second-seeds} through
\[
\left(u_2,\frac{u_2+2v_2}{\sqrt3}\right)
=\left(u_2+v_2,\frac{u_2-v_2}{\sqrt3}\right)O_2,
\]
where
\begin{equation}\label{w4:eq:second-recoupling}
O_2=\frac12\begin{pmatrix}1&\sqrt3\\\sqrt3&-1\end{pmatrix},
\qquad
O_2\binom{\sqrt{3/5}}{1/\sqrt5}
=\binom{\sqrt{3/5}}{1/\sqrt5}.
\end{equation}
Like \(O_1\), the real orthogonal symmetric matrix \(O_2\) maps source
coefficients to target coefficients and fixes the chosen vector.
Equivalently,
\begin{equation}\label{w4:eq:second-seed-match}
\sqrt3(u_2+v_2)+\frac{u_2-v_2}{\sqrt3}
=\sqrt3u_2+\frac{u_2+2v_2}{\sqrt3}.
\end{equation}
The matrices \(O_1,O_2\) express the same seeds in different
coupling bases. The multiplicity rotation
\(O_{\mathrm{merge}}\) is the physical operation included in
\(G_1\) to combine the two equivalent copies.

Concrete left embeddings identify the two target branches:
\(\bigl(A_{\bar3}^{\bar3,8}\bigr)^*:V\hookrightarrow V\otimes W_8\) for left \(3\) and
\(\bigl(A_6^{\bar3,8}\bigr)^*:\operatorname{Sym}^2(V^*)\hookrightarrow
V\otimes W_8\) for left \(\bar6\). Their fixed-basis matrices are the complex conjugates of the initial embeddings in Eq.~\eqref{w4:eq:initial-embeddings}.
{With input column \(j\) and} the symmetric ket
\(|Y\rangle=\sum_{k,c=0}^2Y_{kc}|kc\rangle\),
{the vectorized matrix coordinates are}
\begin{equation}\label{w4:eq:final-embeddings}
\begin{aligned}
\Bigl[(I_{\mathsf S}\otimes{A_8^{3,\bar3}})\bigl(A_{\bar3}^{\bar3,8}\bigr)^*\Bigr]_{abc{;j}}
&=\sqrt{\frac38}
\left(\delta_{ac}{\delta_{bj}}-\frac13\delta_{bc}{\delta_{aj}}\right),\\
\Bigl[(I_{\mathsf S}\otimes{A_8^{3,\bar3}})\bigl(A_6^{\bar3,8}\bigr)^*|Y\rangle\Bigr]_{abc}
&=\frac1{\sqrt2}\sum_{t=0}^2\epsilon_{abt}Y_{tc}.
\end{aligned}
\end{equation}
Hermiticity of the basis matrices gives
\(\bigl(A_8^{3,\bar3}\bigr)^*=\mathrm{SWAP}\,{A_8^{3,\bar3}}\), so complex conjugation
exchanges the two matrix indices in the real tensor formulas of
Eq.~\eqref{w4:eq:initial-embeddings}.
Since \(R(g)\) is real in this basis, the orthogonality, isometry,
and intertwining relations follow by conjugating those of
\({A_{\bar3}^{\bar3,8}}\) and \({A_6^{\bar3,8}}\).
Their seed matches are
\begin{equation}\label{w4:eq:final-seed-identifications}
u_2=\Bigl(\bigl(A_{\bar3}^{\bar3,8}\bigr)^*\otimes I_{\mathsf R}\Bigr){\Phi_{3,\bar3}}(I),
\qquad
\frac{u_2+2v_2}{\sqrt3}
=\Bigl(\bigl(A_6^{\bar3,8}\bigr)^*\otimes I_{\mathsf R}\Bigr){\Phi_{\bar6,\bar3}}(I).
\end{equation}
The full sources of \(\bigl(A_{\bar3}^{\bar3,8}\bigr)^*\otimes I_{\mathsf R}\) and
\(\bigl(A_6^{\bar3,8}\bigr)^*\otimes I_{\mathsf R}\) are \(V\otimes(V^*)_{\mathsf R}\)
and \(\operatorname{Sym}^2(V^*)\otimes(V^*)_{\mathsf R}\),
of dimensions \(9\) and \(18\), respectively.
For the second equality, use the tensor entries of
\((A_{\bar6}^{\bar3,\bar3}\otimes I_{\mathsf R}){\Phi_{\bar6,\bar3}}(I)\)
in the natural factor order:
\([(A_{\bar6}^{\bar3,\bar3}\otimes I_{\mathsf R}){\Phi_{\bar6,\bar3}}(I)]_{kl,d;j}
=(\delta_{dk}\delta_{lj}+\delta_{dl}\delta_{kj})/(2\sqrt2)\).
The contraction is explicitly
\[
\frac1{\sqrt2}\sum_{t=0}^2\epsilon_{abt}
[(A_{\bar6}^{\bar3,\bar3}\otimes I_{\mathsf R}){\Phi_{\bar6,\bar3}}(I)]_{tc,d;j}
=\frac{\epsilon_{abd}\delta_{cj}+\epsilon_{abj}\delta_{cd}}4,
\]
which equals the vectorized entries of \((u_2+2v_2)/\sqrt3\)
by Eq.~\eqref{w4:eq:epsilon-recoupling}.

These seed identifications determine the final decoding. We apply
\(\bigl((A_{\bar3}^{\bar3,8})^*\bigr)^\dagger\otimes I_{\mathsf R}\) and
\(\bigl((A_6^{\bar3,8})^*\bigr)^\dagger\otimes I_{\mathsf R}\),
retain the latter output in \(\operatorname{Sym}^2(V^*)\otimes V^*\),
and map the independent \((3,6)\) branch to the first intermediate summand. These isometries,
followed by the chosen fixed embedding \(A_{\mathcal M_4}^{\mathcal H}:\mathcal M_{4}\hookrightarrow\mathcal H\),
specify a decoding \(\mathcal D_2\) between two 
\(45\)-dimensional physical subspaces. Explicitly, {on the three summands of \(\mathcal M_4\),}
\begin{equation}\label{w4:eq:second-decoding}
\mathcal D_2{\begin{pmatrix}
I_{18}&0&0\\
0&(A_6^{\bar3,8})^*\otimes I_{\mathsf R}&(A_{\bar3}^{\bar3,8})^*\otimes I_{\mathsf R}
\end{pmatrix}}
=A_{\mathcal M_4}^{\mathcal H}.
\end{equation}
The three input images have dimensions \(18,18,9\) and specify
the full source of \(\mathcal D_2\). The embedding exchanges
the second summand's \(6,3\) factors into the \(3,6\) order of its
physical coordinates. In representation coordinates, its seed is
\begin{equation}\label{w4:eq:middle-seed}
\frac1{\sqrt5}{\Phi_{3,6}}(I)\oplus
\frac1{\sqrt5}{\Phi_{\bar6,\bar3}}(I)\oplus
\sqrt{\frac35}{\Phi_{3,\bar3}}(I)=\Phi_{1/5}(I).
\end{equation}
Its representation also equals that of \(\Phi_{1/5}\), namely
\((3,6)\oplus(\bar6,\bar3)\oplus(3,\bar3)\), in the natural left--right
factor order. Seed reconstruction therefore gives
the arbitrary-\(U\) output {map}
\begin{equation}\label{w4:eq:two-query-state}
\frac1{\sqrt5}{\Phi_{3,6}}(U)\oplus
\frac1{\sqrt5}{\Phi_{\bar6,\bar3}}(U)\oplus
\sqrt{\frac35}{\Phi_{3,\bar3}}(U)
=\Phi_{1/5}(U).
\end{equation}
Thus the physical output {map} of \(\mathcal D_2\), defined in
Eq.~\eqref{w4:eq:second-decoding}, is
\(A_{\mathcal M_4}^{\mathcal H}\Phi_{1/5}(U)\).

\subsubsection{Completing the fixed gates on the workspace}
\label{w4:sec:fixed-gate-completion}

The three prescribed isometries now determine the fixed gates. The first, \(G_0\), takes
\(\mathcal H_{\mathsf S}\otimes\mathcal H_{\mathsf P}\) into
\(\mathcal H_{\mathsf S}\otimes\mathcal M_1\). It combines the transpose
\(\mathsf T\) on the second catalyst factor, the initial left decomposition,
and the first right encodings into one fixed isometry.
Let \(\mathcal O_{\mathrm{merge}}\) act as \(O_{\mathrm{merge}}\otimes I_{24}\)
on the multiplicity factor of the two equivalent \(W_8\otimes V\)
copies, and as the identity on the other decoded blocks.
The first decoding has the 99-dimensional target given above,
and its rotated state family lies in
\(\mathcal M_2\oplus\{0\}_{W_8\otimes V}\).
Let \(P_{\mathcal M_2}\) be the coordinate projection from this
99-dimensional target onto the retained \(75\)-dimensional
carrier \(\mathcal M_2\), removing the final zero copy.
Consequently we prescribe \(G_1\) on the source subspace
\begin{equation}\label{w4:eq:middle-gate-source}
G_1=\mathcal E_2P_{\mathcal M_2}\mathcal O_{\mathrm{merge}}\mathcal D_1,
\qquad {\text{on }}
(\mathcal O_{\mathrm{merge}}\mathcal D_1)^\dagger
\bigl(\mathcal M_2\oplus\{0\}_{W_8\otimes V}\bigr),
\end{equation}
where \(\mathcal D_1\) and \(\mathcal E_2\) are defined in
Eqs.~\eqref{w4:eq:first-decoding} and~\eqref{w4:eq:second-encoding},
respectively, and \(\mathcal O_{\mathrm{merge}}\) extends the
multiplicity rotation in Eq.~\eqref{w4:eq:multiplicity-rotation}.
This composition is isometric on the displayed source subspace. The prescribed
\(G_2\) is \(\mathcal D_2\) on the three allowed second-query
blocks, whose total dimension is \(18+18+9=45\), where
\(\mathcal D_2\) is defined in Eq.~\eqref{w4:eq:second-decoding}.

All these source and target subspaces are contained in the same
\(\mathcal H=\mathcal H_{\mathsf S}\otimes\mathcal H_{\mathsf A}\),
with dimensions \((3,90)\). In particular,
\(\mathcal M_3\) occupies a fixed \(30\)-dimensional subspace of
the physical memory. The joint space
\(\mathcal H_{\mathsf S}\otimes\mathcal M_3\) occupies a fixed
\(90\)-dimensional subspace of \(\mathcal H\).
Each prescribed isometry extends to a unitary gate
\(G_\ell\), \(\ell=0,1,2\), on all of \(\mathcal H\) by completing
orthonormal bases of its source and image. Because the prescribed maps
use fixed tensors, phases, and numerical coefficients, the resulting gates
are independent of \(U\). The evolving family stays in their prescribed
domains throughout the circuit. Its endpoint therefore follows from these
actions for any unitary completion, proving
Eq.~\eqref{w4:eq:two-query-endpoint} and Lemma~\ref{w4:lem:two-query}.
\end{wproof}

\subsection{The fixed flip}

\begin{wlemma}[Flip of the intermediate state]\label{w4:lem:middle-flip}
On the intermediate space of Eq.~\eqref{w4:eq:middle-space}
\begingroup\postdisplaypenalty=10000
\[
\mathcal M_{4}=(V\otimes\operatorname{Sym}^2V)
\oplus(\operatorname{Sym}^2(V^*)\otimes V^*)
\oplus(V\otimes V^*),
\]\endgroup
define the fixed unitary
\[
J{(|x\rangle\oplus|y\rangle\oplus|z\rangle)}
={\mathrm{SWAP}_{6,3}|y\rangle
\oplus\mathrm{SWAP}_{3,6}|x\rangle
\oplus(-\mathrm{SWAP}|z\rangle)}.
\]
Its physical extension is
\begingroup\postdisplaypenalty=10000
\[
\widehat J=A_{\mathcal M_4}^{\mathcal H} J\bigl(A_{\mathcal M_4}^{\mathcal H}\bigr)^\dagger+
\Bigl(I_{\mathcal H}-A_{\mathcal M_4}^{\mathcal H}\bigl(A_{\mathcal M_4}^{\mathcal H}\bigr)^\dagger\Bigr),
\]\endgroup
which acts as the identity on the orthogonal complement of
\(A_{\mathcal M_4}^{\mathcal H}(\mathcal M_{4})\)
in \(\mathcal H\).
The abstract flip satisfies
\begin{equation}\label{w4:eq:middle-flip}
J\Phi_\alpha(U)=\Phi_\alpha(U^\dagger)U^\top,
\qquad 0\le\alpha\le\frac12.
\end{equation}
\end{wlemma}
\begin{wproof}
Write \(\rho=\rho_1\oplus\rho_2\oplus\rho_3\) for the representation
on \(\mathcal M_{4}\), with the component actions in Eq.~\eqref{w4:tc:phi-actions}.
They give
\(J\rho(g,h)=\rho(h^*,g^*)J\).
The equal first two coefficients,
\({\Phi_{\bar6,\bar3}}(I)=\mathrm{SWAP}_{3,6}{\Phi_{3,6}}(I)\), and
\(-\mathrm{SWAP}{A_{\bar3}^{3,3}}={A_{\bar3}^{3,3}}\) give
\(J\Phi_\alpha(I)=\Phi_\alpha(I)\).
Proposition~\ref{w4:lem:tc-flip} therefore yields Eq.~\eqref{w4:eq:middle-flip}.
Since \(\widehat JA_{\mathcal M_4}^{\mathcal H}=A_{\mathcal M_4}^{\mathcal H} J\), the actual physical midpoint {map}
after the flip is
\(A_{\mathcal M_4}^{\mathcal H}\Phi_{1/5}(U^\dagger)U^\top\).
\end{wproof}

\subsection{Decoding and catalyst restoration}
\label{w4:sec:four-query-closure}

\begin{wproof}[Proof of Theorem~\ref{w4:thm:S1}]
Two more forward queries decode the flipped {map}.
Lemmas~\ref{w4:lem:two-query} and~\ref{w4:lem:middle-flip} identify it
{as \(\mathcal E(U^\dagger)(U^\top\otimes|c_{4/5}(U^\dagger)\rangle_{\mathsf A})\).}
We therefore apply the inverse circuit
\begin{equation}\label{w4:eq:reverse-circuit}
\mathcal E(U^\dagger)^\dagger
=G_0^\dagger(U_{\mathsf S}\otimes I_{\mathsf A})
G_1^\dagger(U_{\mathsf S}\otimes I_{\mathsf A})G_2^\dagger,
\end{equation}
which uses two further forward queries to \(U\).

{Evaluating the forward maps at \(U^\dagger\) and composing on the right with \(U^\top\) gives the reverse trajectory. The rightmost product reduces to \((U^\dagger)^\top U^\top=I_{\mathsf S}\).}
After \(G_2^\dagger\), the {map} before the third query is
\begin{equation}\label{w4:eq:third-query-before}
\frac1{\sqrt5}(U^\dagger\otimes I_6)
{\Phi_{3,6}}(I)
\oplus
(U_{\mathsf S}^\dagger\otimes R(U^\dagger)\otimes I_{\mathsf R})
\Xi.
\end{equation}
The third forward query cancels the first \(U^\dagger\) in each
summand, leaving
\begin{equation}\label{w4:eq:third-query-after}
\frac1{\sqrt5}{\Phi_{3,6}}(I)
\oplus
(I_{\mathsf S}\otimes R(U^\dagger)\otimes I_{\mathsf R})
\Xi.
\end{equation}
Next \(G_1^\dagger\) reverses the second encoding, multiplicity
rotation, and first decoding. The {map} before the fourth query is
\begin{equation}\label{w4:eq:fourth-query-before}
(U_{\mathsf S}^\dagger\otimes U_{\mathsf L}^\top\otimes I_{\mathsf R}) \left(\frac{v_1}{\sqrt5}+\frac{w_1}{\sqrt{10}}\right) \oplus\frac1{\sqrt5} (U_{\mathsf S}^\dagger\otimes(\operatorname{Sym}^2U^\dagger)_{\mathsf L}\otimes I_{\mathsf R}) \chi_6 \oplus\frac1{\sqrt2} (U_{\mathsf S}^\dagger\otimes\overline{15}(U^\dagger)_{\mathsf L}\otimes I_{\mathsf R}) \chi_{\overline{15}}.
\end{equation}
The fourth forward query replaces each leading \(U_{\mathsf S}^\dagger\)
in Eq.~\eqref{w4:eq:fourth-query-before} by
\(U_{\mathsf S}U_{\mathsf S}^\dagger=I_{\mathsf S}\).
Its output is therefore
\[
\begin{aligned}
&(I_{\mathsf S}\otimes U_{\mathsf L}^\top\otimes I_{\mathsf R}) \left(\frac{v_1}{\sqrt5}+\frac{w_1}{\sqrt{10}}\right) \oplus\frac1{\sqrt5} (I_{\mathsf S}\otimes(\operatorname{Sym}^2U^\dagger)_{\mathsf L}\otimes I_{\mathsf R}) \chi_6 \oplus\frac1{\sqrt2} (I_{\mathsf S}\otimes\overline{15}(U^\dagger)_{\mathsf L}\otimes I_{\mathsf R}) \chi_{\overline{15}}\\
&=G_0\bigl(U^\top \otimes|c_{4/5}(U^\dagger)\rangle_{\mathsf A}\bigr).
\end{aligned}
\]
The last equality follows from Eq.~\eqref{w4:eq:first-query-before} with
\(U\mapsto U^\dagger\){ and right composition by \(U^\top\)}.
The final fixed gate then removes \(G_0\), leaving
\(U^\top\otimes
|c_{4/5}(U^\dagger)\rangle_{\mathsf A}\).

A swap of the adjoint factors restores the catalyst. The vectorization
convention and \(R(U^\dagger)=R(U)^\top\) give the fixed restoration
\(G_{{\mathrm{swap}}}=(1\oplus\mathrm{SWAP}_{8,8})\oplus I_{25}\).
Explicitly,
\[
\mathrm{SWAP}_{8,8}|R(U^\dagger)\rangle\!\rangle
=\sum_{r,s=0}^7R(U^\dagger)_{rs}|s\rangle\otimes|r\rangle
=\sum_{r,s=0}^7R(U)_{sr}|s\rangle\otimes|r\rangle
=|R(U)\rangle\!\rangle.
\]
On the physical memory, this gives
\[
\begin{aligned}
G_{{\mathrm{swap}}}|c_{4/5}(U^\dagger)\rangle_{\mathsf A}
&=\left(\frac1{\sqrt5}\oplus
\frac1{\sqrt{10}}\mathrm{SWAP}_{8,8}
|R(U^\dagger)\rangle\!\rangle\right)\oplus0_{25}\\
&=\left(\frac1{\sqrt5}\oplus
\frac1{\sqrt{10}}|R(U)\rangle\!\rangle\right)\oplus0_{25}
=|c_{4/5}(U)\rangle_{\mathsf A}.
\end{aligned}
\]
Combining encoding, flip, reverse decoding, and restoration gives
\begin{equation}\label{w4:eq:w4-circuit}
\boxed{{C}_4(U)=(I_{\mathsf S}\otimes G_{{\mathrm{swap}}})
\mathcal E(U^\dagger)^\dagger \widehat J\mathcal E(U).}
\end{equation}
The circuit uses exactly four forward processing queries and returns the
complete catalyst, as required by Eq.~\eqref{w4:eq:w4-target}.
Lemma~\ref{w4:lem:catalyst-preparation} prepares that catalyst with three
forward queries before the data arrive. Because the endpoint is a linear-map
identity on all three data columns, tensoring it with the identity on an
external reference also proves the result for entangled data inputs.
\end{wproof}

\begin{wremark}\label{w4:rem:general-catalyst-family}
The construction extends to every singlet--adjoint catalyst
\(|c_a(U)\rangle_{\mathsf P}\) of Eq.~\eqref{w4:eq:catalyst} with
\(4/5\le a\le1\), retaining three forward queries for preparation
and four forward queries for catalytic transposition.
The midpoint has the following left--right representation blocks:
\[
(3,6),\quad(\bar6,\bar3),\quad(3,\bar3),\quad
(15,6),\quad(\bar6,\overline{15}),\quad(15,\overline{15}).
\]
Here \((\lambda,\mu)\) means the action
\(\lambda(g)\otimes\mu(h)\), and bars denote conjugate representations.
Their dimensions are \(18,18,9,90,90,225\), and the normalized seed maps
have squared amplitudes, in this order,
\[
(\alpha,\alpha,3\alpha,3\delta,3\delta,\delta),\qquad
\alpha=\frac{12-7a}{32},\quad\delta=\frac{5a-4}{32},\quad
5\alpha+7\delta=1.
\]
In compatible fixed bases, the flip exchanges blocks one and two and
blocks four and five, and acts as minus the factor swap on blocks three
and six. At \(a=4/5\), \(\alpha=1/5\) and \(\delta=0\), so the last three
blocks disappear.
In this continuous extension, the intermediate representation carrier
has dimension \(450\) for every \(a>4/5\), whereas the endpoint
\(a=4/5\) reduces to the \(45\)-dimensional carrier used above.

For each known \(a\in[4/5,1]\), the same encode--flip--decode mechanism, with gates chosen for that parameter, gives the
four-query endpoint for every \(U\in SU(3)\).
The exact parameter-family endpoints and their assumptions are identified
in the Lean verification entry of Sec.~\ref{sm:lean}; the finite numerical
checks are described separately in Sec.~\ref{sm:symbolic}.
\end{wremark}

The following bounds show that four processing queries are optimal and that, at this threshold, three preparation queries are necessary.

\input{supplemental_storage.tex}

\section{{Computer-assisted analysis}}
\label{sm:mathematica}
\subsection{Symbolic circuit verification}
\label{sm:symbolic}
{The self-contained notebook \nbfile\ }verifies the complete preparation and
processing circuit {in Secs}.~\ref{w4:sec:catalyst-preparation} and
\ref{w4:sec:four-query-closure}. {It constructs} the nine {fixed} $270\times270$
gates as exact sparse arrays{, checks} their two-sided unitarity{, and evaluates}
the three-query preparation and {all three input
columns of} the four-query {processor}. The {companion \texttt{gates.nb} gives}
the {representation construction. Both notebooks and a fresh-kernel replay
script are supplied} in \path{anc/mathematica/}.

For nine independent commuting entries of $U$, write $\operatorname{adj}U$
for the adjugate and set
$R_{rs}^{\rm poly}(U)=\Tr(T_r^\dagger UT_s\operatorname{adj}U)$ and
\[
 |c_{4/5}^{\mathrm{poly}}(U)\rangle
 =\frac{\det U}{\sqrt5}\oplus
 \frac{|R^{\rm poly}(U)\rangle\!\rangle}{\sqrt{10}}\oplus0_{25}.
\]
Here $r,s=0,\ldots,7$ and $0_{25}{\in}\mathbb C^{25}$, {so this} vector {belongs
to} the physical memory $\mathcal H_{\mathsf A}$. The verified degree-seven
identity is
\begin{equation}
 C_4(U)(I_{\mathsf S}\otimes|c_{4/5}^{\mathrm{poly}}(U)\rangle_{\mathsf A})
 =\det(U)\,[U^\top\otimes|c_{4/5}^{\mathrm{poly}}(U)\rangle_{\mathsf A}].
 \label{sm:nb-identity}
\end{equation}
All 810 entry differences vanish exactly in the Mathematica 13.3.1 {replay}.
On $\SU(3)$, $\operatorname{adj}U=U^\dagger$ and $\det U=1$, {giving} the
physical catalytic identity for {every} data {input}. {The} calculation supports
the analytic proof of the fixed $a=4/5$ circuit.

{Running} \texttt{python run.py} in \path{anc/family/} certifies {the} local
Clebsch--Gordan projector identities, recoupling coefficients, and parameter
weights {of} Remark~\ref{w4:rem:general-catalyst-family} exactly {for}
$4/5\le a\le1$. {It} tests the completed {circuits} numerically {at} finitely many
parameter values and $\SU(3)$ matrices.

\subsection{{Lean} verification {of \ref{w4:chap:construction}--\ref{sm:generalstorage}}}
\label{sm:lean}
The Lean 4 sources in \path{anc/lean/} formalize the exact mathematical
{results} of Secs.~\ref{w4:chap:construction}--\ref{sm:generalstorage} and the
main-text {repeated-use constructions}. The coverage ledger records {their
manuscript correspondence and theorem assumptions}.
{The entry \path{anc/lean/MAIN_RESULTS.md} lists the final instantiated theorems for the fixed circuit, the parameter family, the processing and preparation lower bounds, and repeated use, together with their model assumptions.}

{The constructions include preparation} for $0\le a\le1$, the fixed $a=4/5$
{circuit} on the $270$-dimensional workspace, and the compatible family
$4/5\le a\le1$ {with} a $450$-dimensional intermediate carrier in a
$594=3\times198$-dimensional workspace. {These are construction dimensions,
not minimum-memory bounds}. For each known $a$, the {gates are fixed before}
$U$; the {catalytic} endpoint holds for all $U\in\SU(3)$ and arbitrary finite
{references}, including mixed or entangled joint inputs, and returns the same catalyst.

The {bounds cover} finite causal protocols built from completely positive
instruments, {with} mixed initial states, discarded environments, and changing
finite {workspaces}. The repeated-use {constructions cost} $3+4m$ and
$3+8m$, {allow} known {unitaries} on accessible registers, {preserve} passive
references, return the same program{ state}, and reset the added inversion qutrit.

{Replay} uses Lean 4.35.0-rc3 and the pinned Mathlib revision. It
compiles {the} listed {sources} and audits transitive axiom dependencies,
{rejecting} \texttt{sorryAx} {and allowing only} \texttt{propext},
\texttt{Classical.choice}, and \texttt{Quot.sound}. {Source inventories},
commands, and audit logs {accompany} the coverage {ledger}.

\input{supplemental_auxiliary.tex}

%% file: supplemental_storage.tex
\section{Processing bounds with a consumable program{ state}}
\label{sm:consumable}

{Finite preparation cannot lower the processing cost of exact transposition or inversion below $d+1$ or $d^2-1$ forward queries, respectively, even when processing consumes the program{ state}. We derive these bounds by extending the differential completely positive method of Ref.~\cite[Theorem~1]{odake} to $U$-dependent input program{ states}. The circuits use $U$-independent known gates and forward queries in a fixed causal order. We retain finite purifications, coherent measurement records, and processing ancillas in a fixed finite system $\mathsf P$, giving a unitary description without changing either query count.}

For a returned program state prepared with finitely many queries, the repeated-use argument underlying Ref.~\cite[Theorem~1 and Eq.~(5)]{odake} already yields the processing bound: $m$ repetitions use $p+mq$ queries, while the differential bound at $U=I$ for the composed target is $mN$, with $N=d+1$ for transposition and $N=d^2-1$ for inversion. Thus $p+mq\ge mN$ for every $m$ implies $q\ge N$. The proof below also allows a program state that is consumed after one use, so it retains the responses of both the input and output program states instead of relying on repetition.

\begin{wchangedtheorem}\label{sm:thm:consumed}
Let $d\ge2$, $\dim\mathcal H_{\mathsf S}=d$,
and let $\mathcal H_{\mathsf P}$ be any finite-dimensional program space.
Let $\mathrm{Pre}(U)_{\mathsf P}$ be a $p$-query quantum circuit on
$\mathcal H_{\mathsf P}$ that prepares the program state
\[
 |c_{\mathrm{in}}(U)\rangle_{\mathsf P}
 =\mathrm{Pre}(U)_{\mathsf P}|0\rangle_{\mathsf P}{.}
\]
Suppose the data are processed by
a $q$-query quantum circuit ${C_q(U)}$ on
$\mathcal H_{\mathsf S}\otimes\mathcal H_{\mathsf P}$,
whose joint output satisfies
\begin{equation}\label{sm:consumedidentity}
 {C_q(U)}\bigl(I_{\mathsf S}\otimes
 |c_{\mathrm{in}}(U)\rangle_{\mathsf P}\bigr)
 =f(U)\otimes
 |c_{\mathrm{out}}(U)\rangle_{\mathsf P},
\end{equation}
for every $U\in SU(d)$, where
$f(U)=U^\top$ or $U^\dagger$
and $|c_{\mathrm{out}}(U)\rangle_{\mathsf P}$ is the resulting program state.
Then the processing query counts satisfy
\begin{equation}\label{sm:consumedbounds}
 q_\top\ge d+1,\qquad q_{{\dagger}}\ge d^2-1,
\end{equation}
for transposition and inversion, respectively, independently of $p$.
\end{wchangedtheorem}

\par\noindent\emph{Proof.} 
{In this proof, write $C(U):=C_q(U)$.}
The preparation circuit $\mathrm{Pre}(U)_{\mathsf P}$, the processor
${C(U)}$, and both program states have entries polynomial in the matrix
entries of $U$ on $SU(d)$, and are therefore smooth. For
$\mathrm{Pre}(U)_{\mathsf P}$, ${C(U)}$, and $|c_{\mathrm{in}}(U)\rangle_{\mathsf P}$,
this follows from their finite-query expansions. For the output program{ state},
use
\[
 |c_{\mathrm{out}}(U)\rangle_{\mathsf P}
 =\bigl(\langle0|_{\mathsf S}f(U)^\dagger\otimes I_{\mathsf P}\bigr)
 {C(U)}\bigl(|0\rangle_{\mathsf S}\otimes
 |c_{\mathrm{in}}(U)\rangle_{\mathsf P}\bigr).
\]
Here $f(U)^\dagger$ equals $U$ for inversion and the entrywise conjugate $U^*$ for
transposition; the entries of $U^*$ are cofactors of $U$ because
$\det U=1$.

Write the processor ${C(U)}$ as
\[
 {C(U)}=G_q(U\otimes I_{\mathsf P})G_{q-1}
 \cdots(U\otimes I_{\mathsf P})G_0,
\]
with $q$ oracle factors and $U$-independent known unitary gates
$G_0,\ldots,G_q$. {Fixed swaps route queries through the displayed $d$-dimensional factor and are absorbed into these gates.} For $q=0$, ${C(U)}=G_0$.

Fix $U$, and consider $U_\epsilon=e^{i\epsilon H}U$, with $H$ traceless and
Hermitian, and set
${C'_H(U)}:=\left.\frac{d{C(U_\epsilon)}}{d\epsilon}\right|_{\epsilon=0}$.
Each query has derivative
\[
 \left.\frac{d\bigl(U_\epsilon\otimes I_{\mathsf P}\bigr)}
 {d\epsilon}\right|_{\epsilon=0}
 =i(H\otimes I_{\mathsf P})(U\otimes I_{\mathsf P}).
\]
The product rule expresses ${C'_H(U)}$ as the sum of contributions from
the $q$ queries. Multiplication by ${C(U)}^\dagger$ cancels the circuit
factors after each differentiated query, leaving the response pulled
back through the prefix ending at that query.
For $1\le j\le q$, define the prefix
$\mathcal V_j(U):\mathcal H_{\mathsf S}\to
\mathcal H_{\mathsf S}\otimes\mathcal H_{\mathsf P}$ through the $j$th query by
\[
 \mathcal V_j(U)=(U\otimes I_{\mathsf P})G_{j-1}
 \cdots(U\otimes I_{\mathsf P})G_0
 \bigl(I_{\mathsf S}\otimes|c_{\mathrm{in}}(U)\rangle_{\mathsf P}\bigr).
\]
By construction, each $\mathcal V_j(U)$ is an isometry. For the chosen target $f$ and protocol, define the query response by
\[
 \Lambda_{f}(U,H)=\sum_{j=1}^{q}\mathcal V_j(U)^\dagger
 (H\otimes I_{\mathsf P})\mathcal V_j(U).
\]
The same formula defines a completely positive map on the full
matrix algebra, with
\[
 \Lambda_{f}(U,I)=qI.
\]
The preceding differentiation therefore gives
\[
 -i\bigl(I_{\mathsf S}\otimes\langle c_{\mathrm{in}}(U)|_{\mathsf P}\bigr)
 {C(U)}^\dagger {C'_H(U)}
 \bigl(I_{\mathsf S}\otimes|c_{\mathrm{in}}(U)\rangle_{\mathsf P}\bigr)
 =\Lambda_{f}(U,H).
\]
The input program{ state} is held fixed in this expression, and only ${C(U)}$
is differentiated, so $\Lambda_{f}(U,\cdot)$
accounts only for the $q$ processing queries.

The target response in the input data frame is
\[
 g_U(H):=-i f(U)^\dagger
 \left.\frac{df(U_\epsilon)}{d\epsilon}\right|_{\epsilon=0}
 =\begin{cases}
 H^\top,& f(U)=U^\top,\\
 -H,& f(U)=U^\dagger.
 \end{cases}
\]
This follows from
$U_\epsilon^\top=U^\top e^{i\epsilon H^\top}$ and
$U_\epsilon^\dagger=U^\dagger e^{-i\epsilon H}$.

Differentiate the exact input-output identity and multiply on the left by
the adjoint of its unperturbed logical isometry,
\[
 \bigl(I_{\mathsf S}\otimes\langle c_{\mathrm{in}}(U)|_{\mathsf P}\bigr)
 {C(U)}^\dagger
 =\bigl(f(U)\otimes|c_{\mathrm{out}}(U)\rangle\bigr)^\dagger.
\]
Writing $c'_a[H]$ for the directional derivative of the corresponding
program state along $U_\epsilon$, this yields
\[
 i\Lambda_{f}(U,H)+
 \langle c_{\mathrm{in}}|c'_{\mathrm{in}}[H]\rangle I
 =ig_U(H)+
 \langle c_{\mathrm{out}}|c'_{\mathrm{out}}[H]\rangle I.
\]
The input term follows from unitarity,
${C(U)}^\dagger {C(U)}=I_{\mathsf S}\otimes I_{\mathsf P}$;
the output term follows from the factorized, data-independent program{ state}.
Thus
\[
 \Lambda_{f}(U,H)=g_U(H)+\alpha_U(H)I,
\]
where
\[
 \alpha_U(H)=-i\!\left(
 \langle c_{\mathrm{out}}|c'_{\mathrm{out}}[H]\rangle
 -\langle c_{\mathrm{in}}|c'_{\mathrm{in}}[H]\rangle
 \right).
\]
Normalization makes each inner product purely imaginary, so
$\alpha_U(H)$ is real. Directional differentiation makes it real-linear in
$H$. This formula retains both program{-state} derivatives and requires no
equality between the input and output program{ states}.

For each target at the fixed $U$, represent the real-linear
functional $\alpha_U$ on traceless Hermitian matrices by a Hermitian
matrix, denoted $F_\top$ or $F_{{\dagger}}$. Each representative is determined
up to a multiple of $I$; fix this freedom using $\Lambda_f(U,I)=qI$.
The subscripts $\top$ and ${\dagger}$ label transposition and inversion,
respectively. Since traceless Hermitian matrices together with $I$
span all matrices over the complex numbers, linearity gives
\begin{align}
 \Lambda_{\top}(U,M)&=M^\top+\operatorname{Tr}(F_{\top}M)I,
 &\operatorname{Tr}F_{\top}&=q_\top-1,\label{sm:generalET}\\
 \Lambda_{{\dagger}}(U,M)&=-M+\operatorname{Tr}(F_{{\dagger}}M)I,
 &\operatorname{Tr}F_{{\dagger}}&=q_{{\dagger}}+1.\label{sm:generalEI}
\end{align}

With the query-copy factor first, the Choi convention at the fixed $U$
\[
 \operatorname{Choi}(\Lambda_{f})=\sum_{j,k}|j\rangle\langle k|\otimes
 \Lambda_{f}(U,|j\rangle\langle k|)
\]
gives
\begin{align*}
 \operatorname{Choi}(\Lambda_{\top})&=\mathrm{SWAP}_{{d,d}}+F_{\top}^\top\otimes I\succeq0,\\
 \operatorname{Choi}(\Lambda_{{\dagger}})&=F_{{\dagger}}^\top\otimes I-|I_{{d}}\rangle\!\rangle\langle\!\langle I_{{d}}|\succeq0.
\end{align*}

For transposition, $I-\mathrm{SWAP}_{{d,d}}\succeq0$ gives
\[
 0\le\operatorname{Tr}\!\left[(I-\mathrm{SWAP}_{{d,d}})\operatorname{Choi}(\Lambda_\top)\right]
 =(d-1)(\operatorname{Tr}F_\top-d).
\]
Since $d\ge2$, this implies $\operatorname{Tr}F_\top\ge d$.

For inversion, positivity gives
\[
 0\le\langle\!\langle I_{{d}}|\operatorname{Choi}(\Lambda_{{\dagger}})|I_{{d}}\rangle\!\rangle
 =\operatorname{Tr}F_{{\dagger}}-d^2.
\]
Together with $\operatorname{Tr}F_\top=q_\top-1$ and
$\operatorname{Tr}F_{{\dagger}}=q_{{\dagger}}+1$, these inequalities give
$q_\top\ge d+1$ and $q_{{\dagger}}\ge d^2-1$, independently of $p$.\hfill$\square$\par

{For approximate unitary time reversal, Ref.~\cite{chen2026} establishes robust total-query lower bounds. The processing cost alone can vanish when the program{ state} may be supplied with an arbitrary $U$-dependent encoding.}

\begin{wremark}[Approximation with unrestricted program{ states}]\label{sm:rem:unrestricted-approximation}
\normalfont
{Suppose that $U$-dependent program states may be supplied without requiring finite-query preparation. For every $\epsilon>0$, there exist an encoding into a finite-dimensional program space and fixed processors that approximate transposition and inversion using no processing queries, with joint-output fidelity at least $1-\epsilon$ uniformly over all $U\in\mathrm{SU}(d)$ and normalized input states, while returning the program state exactly.}

Indeed, fix $\epsilon>0$ and choose
$0<\delta<\min\{1,\sqrt{\epsilon}\}$.
There is a finite operator-norm $\delta$-net
$\{V_1,\ldots,V_N\}\subset SU(d)$.
For each $U$, select $k(U)$ such that
$\|U-V_{k(U)}\|_{\mathrm{op}}\le\delta$, and set
\[
 |c_{\mathrm{in}}(U)\rangle
 =|c_{\mathrm{out}}(U)\rangle
 =|k(U)\rangle,
 \qquad
 C_0=\sum_{k=1}^{N}f(V_k)\otimes|k\rangle\langle k|.
\]
The processor $C_0$ is a fixed known unitary and uses no oracle queries.
For either $f(U)=U^\top$ or $f(U)=U^\dagger$, it satisfies, uniformly over
all $U$ and normalized data states $|\psi\rangle$,
\[
 \begin{aligned}
 1-\Bigl|
 &\bigl(\langle\psi|f(U)^\dagger
 \otimes\langle c_{\mathrm{out}}(U)|\bigr)\,
 C_0\,
 \bigl(|\psi\rangle\otimes|c_{\mathrm{in}}(U)\rangle\bigr)
 \Bigr|^2
 \le\delta^2<\epsilon.
 \end{aligned}
\]
Moreover, the program{ state} is returned exactly. Thus, even a worst-case
joint-output fidelity requirement permits zero-query processing.
Consequently, no positive processing-query lower bound depending only on
$d$ and $\epsilon$ can hold in this unrestricted program{-state} model; a
nontrivial approximate formulation must also constrain the information
supplied by the program{ state}.
\end{wremark}

\section{Preparation lower bound for qutrit transposition}
\label{sm:generalstorage}

{At the four-query processing threshold of
{Theorem}~\ref{sm:thm:consumed}, exact qutrit transposition requires at
least three preparation queries.}
Our catalytic construction in Sec.~\ref{w4:chap:construction} attains both costs.

\begin{wtheorem}

\label{sm:thm:generalstorage}

{In the model of {Theorem}~\ref{sm:thm:consumed}, suppose }
{that}
{\begin{equation}
 C_q(U)\bigl(I_{\mathsf S}\otimes
 \ket{c_{\mathrm{in}}(U)}_{\mathsf P}\bigr)
 =U^\top\otimes
 \ket{c_{\mathrm{out}}(U)}_{\mathsf P}
 \label{sm:storageidentity}
\end{equation}}
{holds for $q=4$ and every $U\in\SU(3)$,
on an arbitrary finite-dimensional program space $\mathcal H_{\mathsf P}$.
Then preparing the {family of} program{ states}
$\ket{c_{\mathrm{in}}(U)}_{\mathsf P}$ requires at least 3 forward queries.}

\end{wtheorem}

{
\par\noindent\emph{Proof.} {I}t suffices to exclude \(p=2\): a protocol with \(p<2\) can be padded with preparation calls on retained auxiliary qutrits on which the processor acts trivially. {Suppose, for a contradiction, that an exact protocol exists with \(p=2\) and \(q=4\). These assumptions force the first-query prefix to be independent of \(U\), so a fixed encoding would leave only three processing queries, contradicting {Theorem}~\ref{sm:thm:consumed}.}

Replacing the two preparation calls by \(M\in\mathbb C^{3\times3}\) makes \(\lvert c_{\mathrm{in}}(M)\rangle_{\mathsf P}\) homogeneous quadratic. {As in the proof of {Theorem}~\ref{sm:thm:consumed}, l}et \(G_0\) be the fixed joint unitary before the first query on \(\mathsf S\). In a fixed orthonormal basis \(\{\lvert j\rangle_{\mathsf P}\}\), define
\begin{equation}\label{sm:encodingslices}
L_j(M)=(I_{\mathsf S}\otimes\langle j\rvert_{\mathsf P})G_0
\bigl(I_{\mathsf S}\otimes\lvert c_{\mathrm{in}}(M)\rangle_{\mathsf P}\bigr).
\end{equation}

Saturation of the trace inequality in Sec.~\ref{sm:consumable} gives
\(\operatorname{Tr}[(I-\mathrm{SWAP})\operatorname{Choi}(\Lambda_\top)]=0\).
Each query contribution is positive semidefinite, so each is supported on the symmetric subspace. The first contribution is
\(\sum_j\lvert[UL_j(U)]^*\rangle\!\rangle\langle\!\langle[UL_j(U)]^*\rvert\); hence each \(UL_j(U)\) is symmetric. Cubic homogeneity extends this identity from \(\mathrm{SU}(3)\) to \(\mathrm U(3)\). Fixing the two unitary factors in a singular-value decomposition, successive one-variable polynomial identities extend symmetry from unit-modulus diagonal entries to arbitrary complex diagonal entries. {Thus}, for every \(M\in\mathbb C^{3\times3}\),
\begin{equation}\label{sm:generalsymmetry}
ML_j(M)=[ML_j(M)]^\top{.}
\end{equation}

{Completeness of the program basis in Eq.~\eqref{sm:encodingslices}, \(\sum_j\lvert j\rangle\langle j\rvert=I_{\mathsf P}\), and }the unitarity of \(G_0\){ give}
\begin{equation}\label{sm:encodingisometry}
\sum_j L_j(M)^\dagger L_j(M) {=(I_{\mathsf S}\otimes\langle c_{\mathrm{in}}(M)\rvert) G_0^\dagger G_0 (I_{\mathsf S}\otimes\lvert c_{\mathrm{in}}(M)\rangle)} =\bigl\|\lvert c_{\mathrm{in}}(M)\rangle\bigr\|^2I_{\mathsf S}.
\end{equation}

These constraints force the quadratic slices to factor through the adjugate. At a nonzero rank-one \(M=\lvert x\rangle\langle y\rvert\), choose a rank-two \(2\times3\) matrix \(R\) with \(R\lvert x\rangle=0\). Differentiating the symmetry identity in an arbitrary direction \(\delta M\) and multiplying by \(R\) and \(R^\top\) makes \((R\delta M)L_j(M)R^\top\) symmetric. Since \(R\delta M\) is arbitrary, \(L_j(M)R^\top=0\). Choosing $R\delta M$ with only its first or only its second row nonzero forces both columns of $L_j(M)R^\top$ to vanish. The slices annihilate a common two-dimensional subspace, so {Eq.~\eqref{sm:encodingisometry} gives }\(L_j(M){=0}\){ for every \(j\) whenever \(\operatorname{rank}M=1\)}. Indeed, a nonzero vector $v$ in that common kernel gives $0=\|c_{\mathrm{in}}(M)\|^2\|v\|^2$; the scalar is zero, and the positive semidefinite sum in Eq.~\eqref{sm:encodingisometry} then makes every slice zero.

{These rank-one zeros force e}ach quadratic entry {to be} a linear combination of the \(2\times2\) minors: substituting \(M=\lvert x\rangle\langle y\rvert\) pairs its coefficients into the differences \(M_{ab}M_{cd}-M_{ad}M_{cb}\). Equating the coefficients of $x_ax_c\bar y_b\bar y_d$ makes terms with a repeated row or column vanish and pairs the remaining monomials into these minors. Here \(M^\dagger\) denotes the classical adjugate; other daggers denote the usual adjoint. Hence there is a linear map \(\widetilde L_j\) such that
\[
L_j(M)=\widetilde L_j(M^\dagger){.}
\]
 Every invertible \(M\) is the classical adjugate of \(\sqrt{\det M}\,M^{-1}\), with either square root. The symmetry identity and polynomiality thus give
\[
M\widetilde L_j(M)^\top=\widetilde L_j(M)M^\top
\qquad(M\in\mathbb C^{3\times3}).
\]

The same differentiation argument, applied to \(\widetilde L_j(M)^\top\), shows that every column of \(\widetilde L_j(\lvert x\rangle\langle y\rvert)\) lies in \(\operatorname{span}\{\lvert x\rangle\}\). A linear map with this property is right multiplication by a constant matrix: its value at the computational matrix unit \(\lvert a\rangle\langle b\rvert\) has only row \(a\), and evaluation at \((\lvert a\rangle+\lvert c\rangle)\langle b\rvert\) makes that row independent of \(a\). Evaluating the constant at \(M=I\) gives \(\widetilde L_j(M)=ML_j(I)\), and therefore
\begin{equation}\label{sm:encodingfactorization}
L_j(M)=M^\dagger L_j(I).
\end{equation}

For \(U\in\mathrm{SU}(3)\), this factorization gives \(UL_j(U)=L_j(I)\), so
\begin{equation}\label{sm:firstquery-fixed}
(U\otimes I_{\mathsf P})G_0
\bigl(I_{\mathsf S}\otimes\lvert c_{\mathrm{in}}(U)\rangle_{\mathsf P}\bigr)
=G_0\bigl(I_{\mathsf S}\otimes\lvert c_{\mathrm{in}}(I)\rangle_{\mathsf P}\bigr).
\end{equation}

Preparing the fixed state \(\lvert c_{\mathrm{in}}(I)\rangle\) and applying \(G_0\) therefore replaces both preparation queries and the first processing query. The remaining three oracle calls still implement exact transposition, contradicting {Theorem}~\ref{sm:thm:consumed}.
\hfill$\square$\par
}

%% file: supplemental_auxiliary.tex
\subsection{Computational observation on auxiliary dimension}
\label{sm:auxiliary}
{The} construction {in Secs}.~\ref{w4:chap:construction}--\ref{w4:chap:processing}
uses a memory of dimension 90{. Let $M_{\min}$ be the smallest peak} dimension
{outside the query qutrit among} pure coherent, deterministic protocols in
fixed causal order {that use} three data-independent preparation queries and
four processing queries {to transpose an} unknown qutrit unitary exactly {and
return the} catalyst. {The} peak {includes} both stages {and counts} the logical
data whenever they are stored outside the query register. {{Preparation uses \(3^2\cdot8=72\) dimensions outside the queried qutrit, while processing uses \(90\); hence}}
$M_{\min}\le90$.

{Exact coefficient-Gram calculations {motivate} the candidate value}
$M_{\min}=90${. Rational reconstruction} and 1045 nonzero minors modulo 1009{,
checked by an independent elimination, certify injectivity of} the {constructed
four-query coefficient chain}. {A reduction from every admissible physical protocol to this finite coefficient chain, together with the resulting physical-memory rank bound, remains to be established. Thus a matching lower bound valid for all protocols in this model is still open.}

The {attachment \path{anc/aux/} provides a self-contained replay. Running
\texttt{python replay.py} in that directory regenerates} the {rational spaces},
{rebuilds the modular certificate}, {and independently checks every selected
minor using only the Python standard library. Its README states the model,
inputs,} outputs, and {scope of each check}.